\RequirePackage{fix-cm}
\documentclass[smallcondensed]{svjour3}       
\smartqed  
\usepackage{thmtools}
\usepackage{thm-restate}

\usepackage{amsthm}
\usepackage{graphicx}
\usepackage{amsmath,amssymb}
\usepackage{epsfig}
\usepackage{url}
\usepackage{hyperref}
\usepackage{tikz}
\usetikzlibrary{shadows}
\usepackage{xcolor}
\usepackage[framemethod=TikZ]{mdframed}
\usepackage{tkz-berge,float}
\usepackage{enumerate}
\usepackage{bbm} 

\usepackage[linesnumbered,noend,boxruled,nofillcomment]{algorithm2e}

\newcommand{\hide}[1]{}

\newcommand{\spara}[1]{\smallskip\noindent{\bf #1}}

\newtheorem{myexample}{Example}

\DeclareMathOperator*{\argmin}{arg\,min}
\DeclareMathOperator*{\argmax}{arg\,max}

\newcommand{\squishlist}{
 \begin{list}{$\bullet$}
  {  \setlength{\itemsep}{0pt}
     \setlength{\parsep}{3pt}
     \setlength{\topsep}{3pt}
     \setlength{\partopsep}{0pt}
     \setlength{\leftmargin}{2em}
     \setlength{\labelwidth}{1.5em}
     \setlength{\labelsep}{0.5em}
} }
\newcommand{\squishlisttight}{
 \begin{list}{$\bullet$}
  { \setlength{\itemsep}{0pt}
    \setlength{\parsep}{0pt}
    \setlength{\topsep}{0pt}
    \setlength{\partopsep}{0pt}
    \setlength{\leftmargin}{2em}
    \setlength{\labelwidth}{1.5em}
    \setlength{\labelsep}{0.5em}
} }

\newcommand{\squishdesc}{
 \begin{list}{}
  {  \setlength{\itemsep}{0pt}
     \setlength{\parsep}{3pt}
     \setlength{\topsep}{3pt}
     \setlength{\partopsep}{0pt}
     \setlength{\leftmargin}{1em}
     \setlength{\labelwidth}{1.5em}
     \setlength{\labelsep}{0.5em}
} }

\newcommand{\squishend}{
  \end{list}
}

\newif\ifhideproofs

\newcommand{\fullversion}[1]{\ifhideproofs \else {#1} \fi}

\makeatletter
\newcommand{\eqnum}{\refstepcounter{equation}\textup{\tagform@{\theequation}}}
\makeatother

\newcommand{\cov}[2]{{#1}[{#2}]}
\newcommand{\scov}[1]{{#1\uparrow\,}}
\newcommand{\rscov}[1]{{#1\downarrow\,}}
\newcommand{\ineq}[1]{{#1}_{\neq}}

\newcommand{\naturals}{{\mathbb N}}

\providecommand{\poly}{{\operatorname{poly}}}

\newcommand{\setnosp}[1]{\{{#1}\}}
\makeatletter
\newcommand{\set}{\@ifstar
                     \setnosp%
                     \setnosp%
}
\makeatother

\newcommand{\mycomment}[1]{}

\newcommand{\rst}[1]{{\ensuremath{{\mathbin|}\raise-.9ex\hbox{${\scriptstyle{#1}}$}}}}
\newcommand{\subst}[1]{{\ensuremath{\raise-.9ex\hbox{${\scriptstyle{#1}}$}}}}

\DeclareMathOperator*{\expect}{\mathbb{E}}

\DeclareRobustCommand{\Var}[1][]{\ensuremath{\mathrm{Var[#1]}}}
\renewcommand{\epsilon}{\varepsilon}

\DeclareRobustCommand{\calA}[0]{{\mathcal A}}

\DeclareRobustCommand{\calD}[0]{{\mathcal D}}

\DeclareRobustCommand{\calI}[0]{{\mathcal I}}

\DeclareRobustCommand{\calM}[0]{{\mathcal M}}

\DeclareRobustCommand{\calS}[0]{{\mathcal S}}

\DeclareRobustCommand{\calU}[0]{{\mathcal U}}

\journalname{}
\begin{document}
\title{Fair-by-design matching}
\author{David Garc\'ia-Soriano\and Francesco Bonchi}
\institute{ISI Foundation, Italy}
\date{}

\maketitle

\begin{abstract}
Matching algorithms are used routinely to match donors to recipients for solid organs transplantation, for the assignment of medical residents to hospitals, record linkage in
databases, scheduling jobs on machines, network switching, online advertising, and image recognition, among others. Although many optimal solutions may exist to a given matching
problem, when the elements that shall or not be included in a solution correspond to individuals, it becomes of paramount importance that the solution be selected fairly.

In this paper we study individual fairness in matching problems.
Given that many maximum matchings may exist, each one satisfying a
different set of individuals, the only way to guarantee fairness is through randomization. Hence we introduce the \emph{distributional maxmin fairness} framework which provides,
          for any given input instance, the strongest guarantee possible simultaneously for all individuals in terms of satisfaction probability (the probability of being matched
                  in the solution). Specifically, a probability distribution over feasible solutions is \emph{maxmin-fair} if it is not possible to improve the satisfaction probability of any individual without decreasing it for some other individual which is no better off.
In the special case of matchings in bipartite graphs, our framework is equivalent to the \emph{egalitarian mechanism} of Bogomolnaia and Mouline.

                  Our main contribution is a polynomial-time algorithm for fair matching building on techniques from minimum cuts, and edge-coloring algorithms for regular
                  bipartite graphs, and transversal theory. For bipartite graphs, our algorithm runs in $O((|V|^2 + |E| |V|^{2/3})\cdot (\log |V|)^2)$ expected time
and scales                  
to graphs with tens of millions of vertices and hundreds of millions of edges.
To the best of our knowledge, this provides the first large-scale implementation of the egalitarian mechanism.
                  Our experimental analysis confirms that our method provides stronger satisfaction probability guarantees than non-trivial baselines.

\end{abstract}

\section{Introduction}
\label{sec:intro}
Decision-making tools relying on data and quantitative measures have become pervasive
in application domains such as education and employment, finance, search and
recommendation, policy making, and criminal justice.
Awareness and concern about the risks of unfair automated decisions is quickly rising,
as it has been argued
 that decisions informed by data analysis could have inadvertent discriminatory effects due to potential bias existing in the data or encoded in automated decisions.
 Several reports~\cite{Obama1,Obama2}
 call for algorithms that are \emph{``fair by design''} and identify \emph{``poorly designed matching systems''} as one of the main flaws of algorithmic decision-making.
The way
to tackle the ensuing ethical and societal issues has garnered the attention
of the research community~\cite{DworkKeynote}.
However,  despite
the fact that
matching mechanisms
lie at the basis of many automated decision systems, the bulk of the research in the area of
algorithmic bias and fairness
has mainly focused on
avoiding discrimination against a sensitive
attribute (i.e., a protected social group)  in supervised machine
learning~\cite{fair_ml}.

\mycomment{
Our work departs from this literature in three main directions. We focus on (1) \emph{individual fairness}, as opposed to group-level fairness;
(2) \emph{bias stemming from the algorithm design itself}, rather than the bias existing
in the input data; 
(3) 
\emph{combinatorial search problems} (where the solution may not be unique and individuals
correspond to elements
to be included in the solution),
instead of supervised learning.
}

Our work departs from this literature in three main directions: (1) we focus on \emph{individual fairness} (as opposed to group-level fairness);
(2) we focus on \emph{bias stemming from the algorithm design itself}, rather than the bias existing
in the input data; 
(3) instead of supervised learning we focus on 
\emph{matching problems}, where the solution may not be unique and individuals
correspond to elements
to be included in the solution.

In this setting, the satisfaction (utility) function of each individual is based on whether the
individual has been selected or not for inclusion. At the very least, two individuals satisfying all
    relevant criteria equally well
(e.g., having the same skill set) should have, in principle,
    the same expected utility; moreover,  individuals having a wider or a more unique skill set (covering relevant criteria that others can't cover),
should reasonably be rewarded with higher expected utility.
This is often not the case as algorithms may be \emph{``biased by design''}:
bias may stem from something as petty as the order in which the algorithm chooses to process the
list of candidates in its main loop (e.g., by irrelevant attributes such as alphabetical order or
application date), or details about the internal workings of the algorithm.
The prototypical example of a ``biased by design'' algorithm (in a rather extreme way) arises in the  context of \emph{stable
matching} (a different problem from the one considered in this paper): the Gale-Shapley algorithm~\cite{GaleShapley} produces a solution which is  always the best for every man and the worst for every woman, among all feasible solutions, despite the existence of another solution which lies provably ``in the middle'' for every man and woman~\cite{geometry_matchings}.

\spara{Algorithmic bias and randomization.} Consider a job-search setting where we have a
certain number of positions and applicants. 
Assume that each applicant has a binary fitting for each of the positions (either she is fit for the job or not) and a binary satisfaction function (either she is selected or not).
This can be modeled as a matching problem in  a bipartite graph.
Unless 
a matching covering simultaneously
all applicants exists, some of them will have to be left out. 
An unselected applicant could notice that there are other matchings (even maximum-size matchings) satisfying
her. However, any deterministic algorithm is programmed to pick a specific one which may not include her:
she might rightfully deem this unfair.

Unlike the Gale-Shapley algorithm,
    whose bias can be simply characterized by  a theorem, 
{for the problems we consider in this paper}
it may be hard to tell in advance which particular individuals a given algorithm
favours. However, the fact that the bias is not easy to pinpoint
does not mean it does not exist, just that we do not know what it is.

Since no single candidate solution satisfying all individuals at the same time can exist in general,
\emph{we turn our attention to randomized algorithms}, which make random choices to pick
from among several valid solutions.

In our job-search example, imagine there is a single open position  and $n$ applicants fit for it.
Intuitively, all applicants are ``equally qualified'' in this case and the fairest solution would choose one of them uniformly at random,
giving each  applicant
a guaranteed satisfaction (matching) probability of $1/n$.
However, as the graph between applicants and jobs grows more complex, it becomes unclear how to proceed, or what properties one should demand of a fair distribution of solutions.
Our next example illustrates why requiring exactly the same satisfaction probability for all individuals would not make for a good definition.
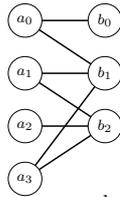
\begin{figure}[t!]
  \begin{center}
\scalebox{0.7}{
\begin{tikzpicture}[scale=0.5]
\begin{scope}[rotate=-90]
    \grEmptyPath[Math,prefix=a,RA=2,RS=0]{4}
    \grEmptyPath[Math,prefix=b,RA=2,RS=3]{3}
\end{scope}
\Edges(b0,a0,b1,a1,b2,a3,b1)
\Edges(a2,b2)
\end{tikzpicture}
}
\vspace{-3mm}
  \caption{An example bipartite graph between people (on the left) and jobs (on the right).
  }\label{fig:bigraph}
\end{center}
\vspace{-1mm}
\end{figure}

\begin{myexample}[Satisfaction probability]\label{example1}
Consider the problem of finding a matching on the bipartite graph of Figure \ref{fig:bigraph} between people (on the left) and jobs (on the right).
Let $\calU = \{ a_0, a_1, a_2, a_3 \}$ and let $\calS$ denote the set of all possible matchings. An individual $u \in \calU$ is satisfied by a solution $S \in \calS$ iff it is matched in $S$ (i.e., she is selected for the job).
Consider the distribution $D$
assigning probability $\frac{1}{3}$ to each of the  following
solutions: $M_1 = \{(a_0,b_0),(a_1,b_1),(a_2,b_2)\}$, $M_2 =
\{(a_0,b_0),(a_1,b_1),(a_3,b_2)\}$, $M_3= \{(a_2,b_2),(a_3,b_1)\}$ and zero probability to all the other matchings.

The satisfaction probability of each individual under distribution $D$ is
exactly the same, namely $\frac23$.
While $D$ might naively look ``fair'',
 notice that the job $b_0$ is left unassigned in $M_3$, despite the existence of a
 fitting candidate  occasionally left unemployed ($a_0$).
 This
artificially restricts the satisfaction probability
of~$a_0$. Observe that, 
for any matching covering a subset $T \subseteq \{a_1, a_2, a_3\}$,
there is another matching covering $T \cup \{a_0\}$. 
So $a_0$ can always be satisfied without impacting anyone else's chances,
hence any reasonable solution should match $a_0$ with probability 1. Other applicants
 will have lower satisfaction probability though (as no matching can satisfy  all of $a_1, a_2$ and $a_3$ at the same time).
\end{myexample}

The insight from Example~\ref{example1} leads us to the key definition of our work.
Our aim is to provide, on any given input instance, the strongest guarantee possible for all individuals, in terms of satisfaction  probability.
We thus introduce the \emph{distributional maxmin fairness} framework. Informally, a distribution over matchings is \emph{maxmin-fair} if it is impossible to improve the satisfaction
probability of any individual without decreasing it for some other
individual which is no better off (see Section~\ref{sec:maxmin} for a formal definition).

\begin{myexample}[Maxmin-fair distribution]\label{example2}
Consider Example~\ref{example1} again. 
A distribution assigning non-zero probability to a solution not covering $a_0$ (such as $M_3$) cannot be maxmin-fair, as otherwise one can increase the satisfaction probability of $a_0$ without detriment to anyone else.
On the other hand, notice that $\{a_1,a_2,a_3\}$ have only two neighbors $\{b_1,b_2\}$, making
it impossible to guarantee satisfaction probability $\ge \frac{2}{3}$ for $a_1, a_2$ and $a_3$  at the same time.
This graph has four maximum matchings:
$M_1$ and $M_2$ from  Example~\ref{example1}, $M_4 = \{ (a_0, b_0), (a_1, b_2), (a_3, b_1) \}$, and
$M_5 = \{ (a_0, b_0), (a_2, b_2),$ $(a_3, b_1) \}$. 
The distribution $F_1$ choosing from among  $M_1, M_2, M_4$ and $M_5$ with probability
$\frac{1}{3},\frac{1}{6},\frac{1}{6},\frac{1}{3}$, respectively, is maxmin-fair.
The satisfaction probabilities of $a_0, a_1, a_2$ and $a_3$ are then
$1, \frac23, \frac23, \frac23$. Any attempt to match, say, $a_1$ with probability 
$> \frac 23$ will necessarily result in satisfaction probability $< \frac 23$ for $a_2$ or
$a_3$.
Another maxmin-fair distribution is, e.g., the distribution $F_2$ choosing uniformly at random
from among $M_1, M_2$ and $M_5$.

\end{myexample}

\spara{Overview of our contributions.}
The contributions of this paper can be summarized as follows:
\begin{itemize}
\item We introduce and characterize the \emph{distributional maxmin-fairness} framework providing, on any given problem instance, the strongest guarantee possible for all
individuals, in terms of satisfaction  probability (Section \ref{sec:maxmin}). While in this paper, for sake of simplicity of presentation,  we focus on matching problems, our
definition applies to a wider variety of problems (such as those listed in Example~\ref{ex:matroids}).

\item In Section \ref{sec:properties}, 
we show  that when the structure of valid solutions forms a matroid (which is the case for matchings),
maxmin-fairness minimizes the
largest inequality gap in satisfaction probabilities between all pairs of individuals, among all Pareto-efficient distributions (Theorem~\ref{thm:maxminsocial}).
We also observe that for such problems
     the ``price of fairness''  is zero: maxmin fairness is
attainable at no cost in solution size.

\item We give a characterization of the ``degree of fairness'' attainable in any bipartite matching instance (Theorem~\ref{thm:matroids}) and any matroid problem instance
(Theorem~\ref{thm:matroids}), generalizing
the classical marriage theorem due to Hall~\cite{hall}.


\item We apply our framework to matching problems in bipartite graphs  (Sections~\ref{sec:problem} and \ref{sec:proofs_improved_algo}), leading to our main contribution:
an exact algorithm for \emph{maxmin-fair bipartite matching} with running
time 
$O((|V|^2 + |E| |V|^{2/3})\cdot (\log |V|)^2)$ (Theorem~\ref{thm:fair_bpm}). 
We also obtain a polynomial time maxmin-fair algorithm for matching in general graphs by a reduction
to the aforementioned bipartite case (Theorem~\ref{thm:simplify}).

\item We discuss how to achieve full transparency for real-world deployment of our framework (Section \ref{sec:transparency}). The discussion leads to the problem of producing a
maxmin-fair distribution with small support, for which we offer an approach making small modifications to our algorithm.

\item Our experiments (Section~\ref{sec:experiments}) show that our algorithm performs faster in practice than
its theoretical running time
and scales to graphs with tens of millions of vertices and hundreds of millions of edges, taking only a few minutes on a simple architecture.
Our analysis confirms that our method provides stronger satisfaction probability guarantees than non-trivial baselines.

\end{itemize}

\ifhideproofs
\fi





\section{Problem definition}
\label{sec:maxmin}

In this section we provide the key definition of our \emph{distributional maxmin-fairness} framework,
considering a very general \emph{search problem} instance $\calI = (\calU,\calS)$
defined over a finite set of  individuals $\calU$ and where
$\calS \neq \emptyset$ denotes the set of feasible solutions for the problem instance $\calI$.
(For example, instance $\calI$ could represent a bipartite graph between jobs and a set $\calU$ of applicants, and $\calS$ the set of all matchings.)
We assume that for every solution $S \in \calS$, each individual $u \in \calU$
    is either fully satisfied  or fully dissatisfied, and this is the only property of the solution
    we are concerned with. Thus, for the sake of simplicity, we will
    identify each solution in $\calS$ with the subset of users satisfied by it, so $\calS \subseteq
    2^{\calU}$.
Note that $\calS$ is defined implicitly by the structure of the problem, and not explicitly encoded in the input. 

Given $\calI$, our problem is to return an element of $\calS$ while providing a  fairness guarantee
to all individuals in $\calU$. Since in general no single candidate solution satisfying all $u \in \calU$ at the same time
exists ($\calU \notin S$), we
seek 
a \emph{randomized algorithm} $\calA$ that, for any given problem instance
$\calI$,  always halts and selects one solution $\calA(\calI)$ from $\calS$.
Thus $\calA$ induces a probability distribution $D$ over $\calS$: $\Pr_D[S] = \Pr[\calA(\calI) = S]$
for each $S \in \calS$.
The \emph{satisfaction probability} of each individual $u \in \calU$  under  $D$ is defined by
$D[u] = \Pr_{S \sim D}[  u \in S ].$





Based on the insight from Example~\ref{example1}, we next provide the key definition of our work.  Informally, a distribution over solutions is \emph{maxmin-fair} if it is impossible to improve the satisfaction probability of any individual without decreasing it for some other individual which is no better off.

 \begin{mdframed}[innerbottommargin=3pt,innertopmargin=3pt,innerleftmargin=5pt,innerrightmargin=5pt,backgroundcolor=gray!10,roundcorner=10pt]
\begin{definition}[Maxmin-fairness]\label{def:maxmin}
A distribution $F$ over $\calS$ is \emph{maxmin-fair} for $\calU$ if for all distributions $D$ over $\calS$ and all $u \in \calU$,
\begin{align}\label{eq:maxmin}
\cov{D}{u} &> \cov{F}{u} \implies \exists v \in \calU \mid
    \cov{D}{v} < \cov{F}{v} \le \cov{F}{u}.
    \end{align}
Similarly, a randomized algorithm is maxmin-fair if it induces a maxmin-fair distribution.
\end{definition}
\end{mdframed}

Finding a maxmin-fair distribution involves solving a continuous optimization problem over (infinitely many) distributions over the set $\calS$ of valid solutions (which is
        commonly exponential in size). The challenge we face is thus how to design an efficient randomized algorithm {inducing} a maxmin-fair distribution.

\begin{problem}\label{prob:main}
For a given search problem,
design a randomized algorithm $\mathcal A$ which always terminates and such that, for each instance  $\calI = (\calU,\calS)$,
 the distribution of $\calA(\calI)$ is maxmin-fair for $\calU$ over $\calS$.
\end{problem}

While our definition applies to a wider variety of search problems,
in this
paper, for sake of simplicity of presentation, we solve Problem~\ref{prob:main} in the case where the search problem is a matching problem in a graph. Let us specify what the sets $\calU$ of users and $\calS$ of
solutions are in this case (Table~\ref{tab:notation} below summarizes the notation used throughout the paper).

\begin{table}[h!]
\begin{center}
\caption{Summary of notation.}
\label{tab:notation}       
\begin{tabular}{ll}
\hline\noalign{\smallskip}
symbol & meaning \\
\noalign{\smallskip}\hline\noalign{\smallskip}
$G = (V, E)$ & undirected, unweighted graph with vertex set $V$ and edge set $E$ \\
$\calU$ & set of users; $\calU \subseteq V$ for matchings \\
$\calS \subseteq 2^{\calU}$ & collection of feasible solutions (possible subsets of satisfied users)\\
$\Gamma_G(A)$  & set of neighbours of  $A \subseteq V$ in $G$ \\
$V = L \dot{\cup} R$ & bipartition of the vertex set $V$ of a bipartite graph \\
$\rho(A)$ & (for graphs) size of the largest matchable subset of $A \subseteq V$ \\
$\rho: 2^L \to \naturals$ & (for matroids) rank function of a matroid with ground set $L$ \\
$D$ & distribution of subsets of $\calS$    \\
$D[v]$ & satisfaction probability of user $u$ under distribution $D$ \\
$D^\uparrow$ & vector of satisfaction probabilities of $D$ in increasing order\\
$D^\downarrow$ & vector of satisfaction probabilities of $D$ in decreasing order\\
$\succ$ & lexicographical order of vectors\\
$\pi(G)$ & minimum satisfaction probability of a maxmin-fair distribution for $G$\\
$\Pi(G)$ & maximum satisfaction probability of a maxmin-fair distribution for $G$\\
$x_{uv}$ & probability of $u$ being matched to $v$ in a fixed maxmin-fair distribution \\
$B_1, \ldots, B_k$ & fair decomposition of $L$ into blocks \\
$ (F(X))_{X \sim D}$ & distribution of random variable $F(X)$ when $X$ is drawn from $D$\\
$G\rst{A}$  & subgraph of $G$ induced by $A \cup \Gamma(A)$ \\
$M\rst{A}$  & restriction of matroid $M$ to the set $A$  \\
$G/A$ &  subgraph of $G$ induced by $(L \cup R) \setminus (A \cup \Gamma(A))$ \\
$M\rst{A}$  & contraction of matroid $M$ to the set $A$  \\
\noalign{\smallskip}\hline
\end{tabular}
\end{center}
\end{table}

Let $G = (V, E)$ be an unweighted simple graph.
A \emph{matching} in $G$ is a set of vertex-disjoint edges of $G$.
A \emph{maximum} matching is a matching of largest size.
The matching $M$
\emph{covers} a vertex $v \in V$
if $v$ is incident to some edge in $M$. 
A set $S \subseteq V$ is \emph{matchable} if there is a matching of $G$ covering all of $S$.
For $S \subseteq V$, define $\rho_G(S)$ as the size of the largest matchable subset of $S$; then
$\rho_G(V)$ is the size of the maximum matching of $G$.
Denote by $\Gamma_G(S)$ the set of neighbours of $S$ in $G$. We will drop the $G$
subscript when no confusion may arise.

In the \emph{fair matching} problem, the input is a graph $G = (V, E)$ and
a set $\calU \subseteq V$ of users. 
Following our assumption of binary satisfaction, user $u \in \calU$ is \emph{satisfied}  by a
matching $M$ if $u$ is covered by $M$. 
The set $\calS$ of valid solutions is  the set  of
matchable subsets of~$\calU$. The set $\calS$ is not part of the input given to the algorithm, but implicitly defined by $G$ and $\calU$.


While the results provided in Section \ref{sec:properties} hold for fair-matching on a general graph, the algorithms developed in Sections~\ref{sec:problem} and \ref{sec:proofs_improved_algo} are for the interesting special case of \emph{one-sided fair bipartite matching} problem, i.e., where
    $G$ is bipartite (with bipartition $V = L\ \dot{\cup}\ R$) and the set of users is given by $\calU=L$.
By solving the one-sided fair bipartite matching problem, we also obtain a polynomial time maxmin-fair algorithm for matching in general graphs by means of a reduction
to the bipartite case (see Theorem~\ref{thm:simplify}).

\section{Fairness and social inequality}
\label{sec:properties}
In this section we present several 
properties of maxmin-fair distributions. These are of independent interest as they provide alternative definitions of maxmin-fairness (Theorems~\ref{thm:lexi}
        and~\ref{thm:maxminsocial}) which are arguably just as natural as Definition~\ref{def:maxmin}; moreover, the latter offers insights into the inequality
distribution properties of maxmin-fairness.
Some results are only stated here; their proofs may be found in Appendix~\ref{sec:proofs_properties}.

\subsection{Basic properties of maxmin-fair distributions}
    An important preliminary observation is that maxmin-fair distributions are unique as far as satisfaction probabilities go, even though several ways may exist to achieve the
    optimal satisfaction probabilities.

\begin{restatable}{lemma}{lemunique}
\label{lem:unique}
Let $F$ and $D$ be two maxmin-fair distributions. Then $\cov{F}{u} = \cov{D}{u}$ for all $u \in \calU$.
\end{restatable}

 \begin{myexample}
In Example~\ref{example2} we gave two maxmin-fair distributions, $F_1$ and $F_2$, which are obtained by combining maximum matchings in different ways, but both satisfy
$\cov{F_1}{a_0} = \cov{F_2}{a_0} = 1$ and $\cov{F_1}{x} = \cov{F_2}{x} = \frac{2}{3}$ for $x \in \{a_1, a_2, a_3\}$.
\end{myexample}

Given a distribution $D$ over $\calS$, write $\scov{D} = (\lambda_1, \ldots, \lambda_n)$ for the
vector of satisfaction probabilities $(D[u])_{u \in \calU}$ sorted in increasing order.
Let $\succ$ denote the lexicographical order of vectors:
$(v_1, \ldots, v_n) \succ (w_1, \ldots, w_n)$
iff there is
 some index
 $i\in[n]$ such that
 $v_j = w_j$ for all $j < i$ and $v_i > w_i$ (the relations $\succeq, \prec$ and $\preceq$ are defined similarly).
%
The following holds.
\begin{restatable}{theorem}{thmlexi}
\label{thm:lexi}
A distribution $F$ is maxmin-fair if and only if $\scov{F} \succeq \scov{D}$ for all distributions $D$ over $\calS$.
\end{restatable}
In other words, a maxmin-fair distribution maximizes the smallest satisfaction probability;
subject to that, it maximizes the second-smallest satisfaction probability, and so on.
\begin{myexample} In Examples~\ref{example1} and ~\ref{example2} we have
$\scov{F_1} = \scov{F_2} = (1, \frac{2}{3}, \frac{2}{3}, \frac{2}{3}) \succ
\scov{D} = (\frac{2}{3}, \frac{2}{3}, \frac{2}{3}, \frac{2}{3}
        )$.
As $\scov{D}$ is not lexicographically maximal, it cannot be maxmin-fair; whereas $\scov{F_1}$ can be shown to be lexicographically maximal, and hence $F_1$ is maxmin-fair.
\end{myexample}

An important observation is that a maxmin-fair distribution always exists for any search problem instance with a feasible solution:
\begin{corollary}Given a search problem instance $\calI = (\calU,\calS)$, a maxmin-fair distribution always exists.
\end{corollary}
\begin{proof}
The probability vectors defining distributions over $\calS$ form a non-empty compact set,
    and the mapping from such vectors to their corresponding sorted satisfaction vectors is
    continuous, so the claim follows from Weirstrass theorem.
\end{proof}


\subsection{Matroid problems}
Theorem \ref{thm:lexi} above provides a definition of maxmin-fairness alternative to Definition~\ref{def:maxmin}.
At the end of this section (Theorem~\ref{thm:maxminsocial})  we provide a second alternative definition characterizing the inequality properties of the distribution of satisfaction probabilities, i.e., the differences between the
satisfaction of the least and most satisfied individuals. It turns out that for a large class of problems, this difference is minimized in a maxmin-fair distribution, making the maxmin-fair distribution the most equitable. However this
does not hold in general for all search problems.

 To be able to state the
class of problems for which Theorem~\ref{thm:maxminsocial} holds, we need to review the concept of \emph{matroids}.
    Many search and optimization problems  can be formulated in terms of
   matroids; they also provide a convenient framework to state and simplify the proofs of some of our results.
\begin{definition}[Matroid problem]\label{def:matroid}
Let $L$ be a finite set. A matroid with ground set $L$ is a non-empty collection $M$
of subsets of $L$ satisfying
the following two properties: (1) 
    if $A\in M$ and $B \subseteq A$, then $B \in M$;
  (2) for any $X \subseteq L$, all maximal subsets of $X$ (with respect to set inclusion) belonging to $M$ have the same size.%
\mycomment
{
\footnote{%
Equivalently, the second condition may be replaced with:
(2')
    if $A, B \in M$ and $|A| < |B|$, then there is $y \in B\setminus A$ such that $A \cup
    \{y\} \in M$.
}
}

A search problem is a \emph{matroid problem}
if for any instance $\calI = (\calU,\calS)$, the set $\calS$ is a matroid.
The elements of a matroid $M$ are called \emph{independent sets}.
The maximal elements of $M$ 
are called \emph{bases}. All bases have
the same size. 
The \emph{rank function} of $M$ is 
$\rho_{M}(S) = \max\{|X| \mid
 X \subseteq S, X \in M
\}$. 


\end{definition}

\begin{myexample}\label{ex:matroids}
The following are matroids (see~\cite{lawler_book}):
\begin{itemize}
    \item The collection of sets of matchable vertices in a graph~\cite{matching_book}. This well-known result follows
    from a theorem of Berge~\cite{berge_matchings} that
        we may extend any matchable set of \emph{vertices} to a matchable set of maximum size.
        {By contrast, the collections of sets of \emph{edges} forming a matching is \emph{not} a matroid.}
    \item The collection of sets of vertices in a graph for which edge-disjoint paths from another single specified vertex exist.
    \item The collection of linearly independent sets of vectors over a finite vector space.
    \item The collection of forests (acyclic sets of edges) in a graph.
\end{itemize}
The search problems corresponding to finding any of the above are matroid problems.
\end{myexample}


\mycomment{
        This is a direct consequence of the following key property of matchings which asserts
        we may extend any matchable set of \emph{vertices} to a matchable set of maximum size.
        \begin{theorem}[Berge~\cite{berge_matchings}, Edmonds\cite{edmonds_algo}]\label{thm:berge}
        For any matchable set $S\subseteq V$, there is a matchable set $T \supseteq S$ with $|T| = \rho(V)$.
        Moreover, given $S$, a matching covering $T$ may be found in polynomial time.
        \end{theorem}
        A direct consequence of Theorem~\ref{thm:berge}
         is the well-known fact (see~\cite{matching_book}) that the sets of \emph{matchable vertices} in a graph form a
        matroid,
           and therefore the
        matching problem in any graph is a matroid problem.
        \footnote{By contrast, the collections of sets of \emph{edges} forming a matching is not a matroid.}

}


 Notice that any set $X$ appearing with non-zero probability in a maxmin-fair distribution must be maximum in size; otherwise, by property (2) in Definition~\ref{def:matroid}, $X$ is not maximal so we could
replace $X$ with some strict superset $Y \supsetneq X$, which can only increase the satisfaction probability of every $u \in L$. It is in this sense that the ``price of fairness''
is zero for matroid problems: the support of a maxmin-fair distribution consists only of solutions of maximum size, so it is never necessary to trade fairness for solution size.
In particular this holds for matching problems as well.


{ 

\subsection{Minmax-fairness}
By definition, maxmin-fair distributions give the highest possible satisfaction probabilities to the worst-off individuals.
    To investigate the inequality properties of these, we introduce a dual notion of minmax-fair distributions, which by contrast give
    the lowest possible satisfaction probabilities to the best-off individuals.  It turns out that for matroid problems both notions coincide, provided that
    we exclude Pareto-inefficient distributions.

\begin{definition}[Pareto efficiency]
A distribution $E$ is (ex-ante) \emph{Pareto-efficient} if there is no distribution $D$ such that
$\cov{D}{u} \ge \cov{E}{u}$ for all $u\in \calU$ and
$\cov{D}{u} > \cov{E}{u}$ for at least one
$u \in\calU$.
\end{definition}
The notion of  Pareto-efficiency expresses the impossibility of improving the
satisfaction probability of some user without detriment to anyone else.
Clearly any
maxmin-fair distribution is Pareto-efficient,
hence any solution in its support
is \emph{maximal} (with regard to set inclusion).

The notion of minmax-fairness outlined above requires that
no user satisfaction can be decreased without increasing that of another user which is no worse off,
   or losing Pareto-efficiency.

\begin{definition}
A Pareto-efficient distribution $F$ over $\calS$ is \emph{minmax-Pareto} (or minmax fair) for $\calU$ if for all
Pareto-efficient distributions $D$ over $\calS$ and all $u \in \calU$, it holds that
{
    \begin{align}
    \nonumber
    \cov{D}{u} &< \cov{F}{u} \implies
    \exists v \in \calU \mid
        \cov{D}{v} > \cov{F}{v} \ge \cov{F}{u}.
        \end{align}
}
\end{definition}
Requiring Pareto-efficiency is redundant for maxmin-fairness, but crucial for
minmax-Pareto efficiency; without it, the definition would be met by a distribution of solutions satisfying 
{nobody} (for example, a solution which always returns the empty matching).

In Appendix~\ref{sec:proofs_properties} we present analogues to Lemma~\ref{lem:unique} and Theorem~\ref{thm:lexi} (Lemma~\ref{lem:unique} and Theorem~\ref{thm:lexi2}) for minmax-fairness.

\mycomment{
    Similarly to Lemma~\ref{lem:unique}, the satisfaction probabilities of a minmax-Pareto distribution are uniquely determined (Lemma~\ref{lem:unique2} in Appendix~\ref{sec:proofs_properties}).
    We also have an analogue of Theorem~\ref{thm:lexi}:
    let $\rscov{F}$ denote the satisfaction vector of distribution $F$ sorted in \emph{decreasing} 
    order.
    Then,  for matroid problems,
    a distribution $F$ is minmax-Pareto if and only if $F$ is Pareto-efficient and $\rscov{F} \preceq \rscov{D}$ for all Pareto-efficient distributions~$D$ (Theorem~\ref{thm:lexi2} in
            Appendix~\ref{sec:proofs_properties}).
}

\subsection{Inequality properties}
The main result of this section, Theorem~\ref{thm:maxminsocial} is that, for matroid problems, the notions of
minmax fairness and maxmin fairness coincide; intuitively,
any excess satisfaction probability for the best-off user can be taken away from him and
redistributed to others.
This also implies that
the maxmin-fair solution
minimizes the largest gap in satisfaction probabilities; among
those, it minimizes the second-largest gap, etc.%

\begin{definition}
\label{def:social}
The \emph{sorted inequality vector} of a distribution $D$ over $\calS$, written
$\ineq{D}^\downarrow$, is the vector of all pairwise differences in the satisfaction probabilities
of the elements of $\calU$ under $D$, sorted in decreasing order.
\end{definition}
\vspace{-0.2cm}
\begin{restatable}{theorem}{thmmaxminsocial}
\label{thm:maxminsocial}
For matroid problems, the following are equivalent:
(1) $D$ is maxmin-fair;
(2) $D$ is minmax-Pareto;
(3) $D$ is Pareto-efficient and 
$\ineq{D}^\downarrow \preceq \ineq{E}^\downarrow$ for all Pareto-efficient distributions $E$ over $\calS$.
\end{restatable}
The proof may be found in Appendix~\ref{sec:proofs_properties}.

Note that this result does not hold in general for non-matroid problems; the following shows a counterexample.
\begin{myexample}
Consider the problem instance where the set of individuals is $\calU = \{0,1,2,3\}$ and the set of feasible solutions is $\calS= \{ \{0,1\}, \{1,3\}, \{0,2,3\} \}$.
Here elements $1$ and $2$ never appear together in a solution, so the minimum satisfaction probability cannot exceed $\frac 12$. In order to achieve $\frac 12$ we need to choose
$\{0,2,3\}$ with probability exactly $\frac 12$; this fixes the satisfaction probabilities of $2$ and $1$ to $\frac 12$, and to maximize the second-smallest probability we need to
pick $\{0,1\}$ and $\{1,3\}$  with probability $\frac 14$ each. 
This is the maxmin-fair distribution $D_1$ and its maximum inequality is $\frac14$.
However, a similar argument shows that
the minmax-fair distribution $D_2$ is different: it uses 
each element of $\calS$ with probability $\frac 13$ and has maximum inequality $\frac 13$.

Note that in this case one may verify that 
$D_1$ still minimizes maximum inequality, but 
by considering the complements of each
element of $\calS$, one can give a similar example where the
maxmin-fair distribution does not minimize inequality.
\end{myexample}

\mycomment{

\begin{myexample}
Consider the problem instance where the set of individuals is $\calU = \{0,1,2,3\}$ and the set of feasible solutions is $\calS= \{ \{0,1\}, \{1,3\}, \{0,2,3\} \}$.
\textcolor[rgb]{1.00,0.00,0.00}{SAY WHY IT IS A COUNTEREXAMPLE. I can't remember, need to do some calculations. I just wrote the example down so I wouldn't forget.  Let's stick to
    the most interesting things.}
\end{myexample}
}

\section{A polynomial-time algorithm for maxmin-fair matching}
\label{sec:problem}
In this section we present our main contribution: a polynomial-time algorithm for maxmin-fair matching. We present our algorithm for the \emph{one-sided fair bipartite matching}
problem. This is the special case of fair matching where:
\begin{itemize}
\item $G$ is bipartite (with bipartition $V = L\ \dot{\cup}\ R$),
\item the set of users is $\calU=L$, 
\item there is a maximum matching covers all the right-side vertices but not all the left-side vertices (i.e., $\rho(L) = |R| < |L|$),
\item there are no degree-0 vertices (which can always be removed).
\end{itemize}
This setting corresponds to the job-search setting that we use in our examples throughout the paper.
We will see later (Section~\ref{sec:transversal}) that the general fair matching problem in non-bipartite graphs, with arbitrary user sets $\calU \subseteq V$ and with no further restrictions, can be reduced to this special case in polynomial time.  
Before presenting the building blocks of our algorithm in full detail, we provide an overview of our techniques.
 Some results are only stated here; their proofs may be found in Appendix~\ref{sec:proofs_problem}.

\subsection{Overview of our techniques}
We next provide an overview of how we obtain our main result: an efficient algorithm for maxmin-fair bipartite matching.

\begin{enumerate}[(1)]
\item The first ingredient (Section \ref{subsec:fairnessparameter}) is a characterization of the \emph{fairness parameter}, i.e., the maximum
satisfaction probability which can be guaranteed for every user. By using Hall's theorem we prove (Corollary~\ref{thm:lambda1}) that the fairness parameter is determined by a ``blocking''
set of  vertices with the smallest neighborhood-to-size ratio. Unfortunately, the proof does not lead to an efficient algorithm to find this set.

\item Thus we proceed to write down a linear program for a fractional variant of the problem (Section \ref{subsec:lp}). Inspired by a technique developed by
Charikar~\cite{charikar} for the \emph{densest-subgraph problem}, we show (Lemma~\ref{lem:lp_bipartite}) that any fractional solution can be leveraged to find a blocking set of
vertices. The neighbors of the blocking set cannot be matched to any vertex outside the blocking set in any maxmin-fair distribution. We use this fact to argue inductively
(Theorem~\ref{thm:decomp}) the existence of a ``fair decomposition'' of the set of left vertices with the following property: vertices on higher levels can be allowed larger satisfaction probabilities, regardless of which edges are used to match the vertices on lower levels.

\item Having computed
the assignment probabilities $x_{uv}$ (the probability of each pair of vertices being matched)
    of some  maxmin-fair distribution within each block in the decomposition,
we can turn each of them into an actual distribution of matchings by finding the Birkhoff-von Neumann decomposition of a 
doubly-stochastic matrix. Then we combine them into a single distribution.

\item To obtain our faster algorithm (which also returns the exact optimal solution), we avoid the use of linear programming and instead present a technique to find several blocks in parallel with a single min-cut computation (Section~\ref{sec:proofs_improved_algo}). We show that a logarithmic number of minimum cut computations
suffice to obtain the fair decomposition in full. Then we argue that given the decomposition and satisfaction probabilities, the required distribution of matchings can be found by coloring
the edges of an appropriately constructed regular bipartite graph, for which task we leverage the fast algorithm of Goel et al.~\cite{perfect_nlogn}.

\end{enumerate}

\subsection{Fairness parameter} \label{subsec:fairnessparameter}
{We next ask the following important question: what is the minimum satisfaction probability $\pi(G)$ of a maxmin-fair distribution for $G$?}
    {
Hall's marriage theorem gives a necessary and sufficient condition
    for the existence of a matching covering the whole of $L$, which is equivalent to having $\pi(G) = 1$.

\begin{theorem}[Hall~\cite{hall}]\label{thm:hall}
In a bipartite graph with bipartition $(L, R)$, the set $L$ is matchable if and only if $|\Gamma(S)|
\ge |S|$ for all $S \subseteq L$.
\end{theorem}
    }
\fullversion{

\mycomment{
\begin{proof}
We show that $$ \min_{T \subseteq S} \rho(T) / |T| = \min_{T \subseteq S} |\Gamma(T)| / |T|. $$
If $\rho(T) \ge \lambda |T|$ for all $T \subseteq S$, then
$$ |\Gamma(T)| \ge \rho(T) \ge \lambda |T| $$ for all $T \subseteq S$.

On the other hand, if $|\Gamma(T)| \ge \lambda |T|$ for all $T \subseteq S$, then
for all
$U \subseteq T$,
    $$ \Gamma(U) + |T| - |U| \ge |T| + \lambda |U| - |U| = |T| - (1 - \lambda) |U| \ge |T| - (1 - \lambda) |T| = \lambda |T|, $$
so
$$ \rho(T) = \min_{U \subseteq T} |\Gamma(U)| + |T| - |U| \ge \lambda |T|. $$
\end{proof}
}
}

We show a generalization of Hall's theorem which will prove useful to characterize the
fairness parameter in bipartite matching.

\begin{restatable}{theorem}{thmalphas}
\label{thm:alphas}
Let $\{ \alpha_v \mid v\in L\}$ be reals in $[0, 1]$.
A necessary and sufficient condition for the existence of  a distribution $D$ of matchings of $G$ such that $\cov{D}{v} \ge \alpha_v$ for all $v \in L$ is

\begin{equation}\label{eq:hall}
 \text{for all $S \subseteq L$,}\quad |\Gamma(S)| \ge \sum_{v \in S} \alpha_v.
 \end{equation}
\end{restatable}
\begin{proof}
Necessity is clear because no matching can cover more than $|\Gamma(S)|$ elements
of any set $S$, but the expected number of elements of $S$ covered by $D$ is $\sum_{v\in S} \cov{D}{v} = \sum_{v\in S} \alpha_v$ by linearity of expectation.

For sufficiency, we may assume that all the $\alpha_v$ are rational
because~\eqref{eq:hall} is a finite set of inequalities with integral coefficients, so the maximizer
of $\sum_v \beta_v$ subject to $|\Gamma(S)| \ge \sum_{v \in S} \beta_v$ and $\beta_v \ge
\alpha_v$ will have $\beta_v \in \mathbb{Q}$.
Let $M$ be a suitable common denominator, so that $\alpha_u = \beta_u = n_u /
M$ where $M\ge n_u \in \naturals$. Construct a graph $G'$ with
\begin{itemize}
    \item $n_u$ replicas $u^{(1)}, \ldots, u^{(n_u)}$ of each $u \in L$;
    \item $M$ replicas $v^{(1)}, \ldots, v^{(m)}$ of each $v \in R$;
    \item $V(G') = L' \cup R'$, where $L' = \{ u^{(i)} \mid u \in L, i \le n_u \}$
and
$R' = \{ v^{(i)} \mid v \in R, i \le M \}$.
    \item $E(G') = \{ (u^{(i)}, v^{(j)}) \mid (u, v) \in E(G), i \le n_u, j \le M\} $.
   \end{itemize}
This graph is bipartite with bipartition $(L', R')$. Notice that vertices with $\alpha_v = 0$ have
        no replica in $G'$.

Consider (in $G$) the sets $A_k = \{ u \in L \mid n_u \ge k \}$ for $k = 1, 2, \ldots, M$.
Given $k$ and a set $S \subseteq A_k$ let
$S^{(k)} = \{ u^{(k)} \mid u \in S \}$. 
If $A_1 = \emptyset$ the theorem is trivial. Otherwise, let $H_1$ denote the subgraph of $G'$ induced by $A_1^{(1)} \cup R'$.
Any subset of $A_1'$ in $H_1$ is of the form $S^{(1)}$,
for some $S \subseteq A_1$. Using~\eqref{eq:hall} we obtain
    $$ |\Gamma_{H_1}(S^{(1)})| = M \cdot |\Gamma_G(S)| \ge \sum_{u \in S} n_u \ge |S^{(1)}|,$$
        because $n_u \ge 1$ for $u \in A_1\supseteq S$.
By Hall's Theorem, there is a matching $X_1$ in $H_1$ covering $A_1^{(1)}$.

If $A_2 \neq \emptyset$, let $H_2$ denote the subgraph of $G'\setminus V(X_1)$ induced by $A_2^{(2)} \cup R'$.
As we removed the edges of the matching $X_1$,
        the number of neighbours in $G'$ of any set $S\subseteq A_2$ has decreased by at most $|S|$,
so for any $S \subseteq A_2$ he have
    $$ |\Gamma_{H_2}(S^{(2)})| \ge M \cdot |\Gamma_G(S)| - |S| \ge \sum_{u \in S} (n_u - 1) \ge
    |S^{(1)}|,$$
        because $n_u \ge 2$ for $u \in A_2\supseteq S$.
Hence there is a matching $X_2$ in $H_2$ covering $A_2^{(2)}$.
Proceeding similarly, 
we obtain a set of vertex-disjoint
matchings in $G'$ such that their union is
   a matching $X'$ in $G'$ covering $L'$.
   By restricting $X'$ to each replica of $R$ in $R'$, we can decompose $X'$ into
   $M$
   matchings $X_1, \ldots, X_M$, each of them inducing a matching in $G$.
   Furthermore, each $u \in L$ is covered
in
exactly $n_u$ of these, since $X'$ covers $L'$. Thus the uniform distribution over $X_1, \ldots, X_M$
yields coverage probability $n_u / M = \alpha_u$ for each $u \in L$.
\end{proof}

The proof gives a maxmin-fair distribution which is uniform over a multiset of $M$ matchings, but
$M$ may be fairly large, as large as $2^{\Omega(\sqrt {|\calU|})}$ in some instances.


\mycomment{
We  state a result dual to Theorem~\ref{thm:alphas} about guaranteed satisfaction probabilities bounded from above. 
\begin{theorem}\label{thm:alphas2}
Suppose $\rho(L) = |\Gamma(L)|$. Let $\{ \alpha_v \mid v\in L\}$ be reals in $[0, 1]$. A necessary and sufficient condition for the
existence of  a Pareto-efficient
distribution $D$ of matchings of $G$ such that $\cov{D}{v} \le \alpha_v$ for all $v \in L$ is
\begin{equation}\label{eq:hall2}
 \text{for all $S \subseteq L$,}\quad |\Gamma(S)| + \sum_{v \notin S} \alpha_v \ge |\Gamma(L)|.
 \end{equation}
\end{theorem}
}


\begin{corollary}\label{thm:lambda1}
   The minimum satisfaction probability in a maxmin-fair distribution for the one-sided bipartite matching problem is
$$\pi(G) = \min \left\{ \frac{|\Gamma(S)|}{|S|} \mid \emptyset \neq S \subseteq L \right\}.$$
\mycomment{
                and
                the maximum coverage probability is
                \begin{equation}\label{eq:lambda2}
                \lambda_n = \min\left(1, \max \left\{ \frac{|\Gamma(L)| - |\Gamma(S)|}{|L\setminus S|}  \mid S\subsetneq L
                        \right\} \right).
                \end{equation}
}
\end{corollary}
\begin{proof}
Fix a parameter $\lambda\in[0,1]$. By Theorem~\ref{thm:alphas}, a distribution with satisfaction probability at least $\lambda$ for all $L$ exists if and only if $\Gamma(S) \ge \lambda  |S|$ for all $S\subseteq V$.
\end{proof}

In Appendix~\ref{sec:proofs_problem} we prove a dual result for the maximum satisfaction probability:
\begin{corollary}\label{thm:lambda2}
   The maximum satisfaction probability in a maxmin-fair distribution for the one-sided bipartite matching problem is
                \begin{equation}\label{eq:lambda2}\nonumber
                \Pi(G) = \max \left\{ \frac{|\Gamma(L)| - |\Gamma(S)|}{|L\setminus S|}  \mid S\subsetneq L \right\}.
                \end{equation}
\end{corollary}                

\subsection{A compact LP formulation for the fairness parameter}\label{subsec:lp} 
{Below we write a linear program for computing $\pi(G)$}. 
\mycomment{
    {
        \small
    \begin{equation}\label{lp:mu}
    \begin{array}{rrclcl}
    \displaystyle \min & \displaystyle \sum_{v \in R} y_v \\
    \textrm{s.t.}
    & y_v - y_u &\ge& 0 & \forall (u, v) \in E \subseteq L \times R \\
    & \displaystyle \sum_{u \in L}\, y_u  & = & 1 & \forall u\in L \\
    & y_u, y_v &\ge& 0 & \forall u \in L, v \in R \\
    \end{array}
    \end{equation}
    }
}
    {
    \begin{equation}\label{lp:mu}
    \begin{array}{rrclll}
    \textrm{minimize} & \sum_{v \in R} y_v \\
    \textrm{s.t.}
    & y_v - y_u &\ge& 0 & \; \; \; \forall (u, v) \in E \subseteq L \times R \\
    & \sum_{u \in L} y_u  & = & 1 & \; \; \;  \forall u\in L\\
    & y_u, y_v &\ge& 0 & \; \; \; \forall u \in L, v \in R \\
    \end{array}
    \end{equation}
    }


Any set $S \subseteq L$ can be represented by a feasible solution to this LP by setting
$y_x = \frac{1}{|S|}$ for all $x \in S \cup \Gamma(S)$.
\begin{restatable}{lemma}{lemconstructsol}
\label{lem:construct_sol}
For any non-empty set $S \subseteq L$, there is a feasible solution to LP~\eqref{lp:mu} with value
$\frac{ |\Gamma(S)| }{|S|}$.
\end{restatable}
\begin{proof}
Define $y_x = \frac{1}{|S|}$ for all $x \in S \cup \Gamma(S)$ and $y_x = 0$ elsewhere. 
Then
 $\sum_{u \in L} y_u = \sum_{u \in S} \frac{1}{|S|} = 1$ and for every edge $(u, v) \in L \times
 R$ we  have either
$y_u = 0$ (in which case $y_v \ge 0 = y_u)$ or $y_u = 1/|S|$; the latter implies $u \in S$ and $v \in \Gamma(S)$, so $y_v = 1 / |S| = y_u$.
This proves feasibility. Finally, $\sum_{v \in R} y_v = \sum_{v \in \Gamma(S)} \frac{1}{|S|} = \frac{| \Gamma(S) |}{|S|}$.
\end{proof}

\fullversion{The following shows how to round an optimal solution LP~\eqref{lp:mu}  to obtain a set $S$ of vertices such that $|\Gamma(S)|/|S|$ equals the optimal value.
    A similar
    technique has been used by Charikar~\cite{charikar} for the densest subgraph LP.}
\begin{lemma}\label{lem:lp_bipartite}
Let $\{y_w\}_{w \in L \cup R}$ be an optimal solution to~\eqref{lp:mu}.
Then the set $S = \{ v \in L \mid y_v > 0 \}\neq \emptyset$ satisfies
$\frac{ |\Gamma(S)| }{|S|} = \sum_{v \in r} y_v$.
\end{lemma}
{
\begin{proof}
Write $\lambda = \sum_{v \in R} y_v$.
For any $r \in (0, 1)$, define $S(r) = \{u \in L \mid y_u \ge r \}$ and $T(r) = \{ v \in R \mid y_v \ge r \}$.
We show that $T(r) = |\Gamma(S(r))|$ and $|T(r)| / |S(r)| = \lambda$ for \emph{every} $r \in (0, 1)$.
To see this, observe that for any $v \in R$, $y_v \ge \max_{u \in \Gamma^{-1}(v)} y_u$. In fact in any optimal solution equality must hold:
$y_v = \max_{u \in \Gamma^{-1}(v)} y_u$
 for all $ v \in R$; otherwise
we may decrease some $y_v$ and hence the objective function without sacrificing feasibility. Consequently,
\begin{align*}
v \in T(r) &\Leftrightarrow y_v \ge r \Leftrightarrow \max_{u \in \Gamma^{-1}(v)} y_u \ge r
\Leftrightarrow\\& \Leftrightarrow \exists u \in \Gamma^{-1}(v) \text{ such that } y_u \ge r \Leftrightarrow
v \in \Gamma(S(r)).
\end{align*}
Recall from Lemma~\ref{lem:construct_sol} that we can construct a solution to LP~\eqref{lp:mu} from any non-empty set.
Since $\lambda$ is the optimal value of LP~\eqref{lp:mu}, for any $r$ for which $S(r) \neq \emptyset$ we have $|T(r)| / |S(r)| \ge \lambda$, i.e., $0 \le |T(r)| - \lambda |S(r)|$.
The latter also holds if $S(r) = \emptyset$. On the other hand, if we pick $r$ uniformly at random from $(0, 1)$, we have
$$ \expect_r [ |S(r)| ] = \sum_u \Pr_r [ u \in S(r) ] = \sum_{u} \Pr_r [ r \le y_u ] = \sum_u y_u = 1, $$
$$ 
\expect_r [ |T(r)| ]= \sum_v \Pr_r [ v \in T(r) ] = \sum_{v} \Pr_r [ r \le y_v] = \sum_v y_v = \lambda, $$
so
$ 0 \le \expect_r [ | T(r) | - \lambda | S(r) | ] = \expect_r [ |T(r)| ] - \lambda \cdot \expect_r [ |S(r)| ] = \lambda - \lambda \cdot 1 = 0,$
which implies that $T(r) - \lambda \cdot S(r) = 0$ almost surely when $r$ is uniform in $(0, 1)$.
Observe that $T(r) / S(r)$ is piecewise-constant in its domain (all distinct possibilities are given by taking $t = y_w$ for some $w \in L \cup R$).
Moreover, for any $r \in (0, 1)$ there is some interval $I$ of non-zero length such that for all $r' \in I$, then $S(r) = S(r')$ and $T(r) = T(r')$.
Thus, any event that is a measurable function of $S(r)$ and $T(r)$ and holds with probability 1 when $r \sim U(0, 1)$ must actually hold for every $r \in (0, 1)$ as well.

Thus, $|T (r) | = \lambda |S (r)|$ for all $r \in (0, 1)$.
In particular if we pick $r_0 = \min_{u \in L} y_v$, then $S(r_0) = \{ v \in L \mid y_v > 0 \}$
satisfies $\sum_{v \in S(r_0)} y_v = 1$, hence is non-empty, and by the above
we have
$ |\Gamma(S(r_0))| - \lambda \cdot |S(r_0)| = 0 ,$
as desired.
\end{proof}
}

    In combination with  Corollary~\ref{thm:lambda1}, these two lemmas yield an effective method of
    computing $\pi(G)$:
\begin{corollary}
In the one-sided fair bipartite matching problem,
the fairness parameter $\pi(G)$  is equal to the optimum value of
the LP {in}~\eqref{lp:mu}.
\end{corollary}

\subsection{Fair decompositions}\label{sec:fair_decomp} 
The next ingredient towards an efficient algorithm is to find a decomposition of $L$ according to different levels of satisfaction probability in the maxmin-fair distribution.
In  Figure~\ref{fig:hugegraph}, the set of left vertices with smallest neighbor-to-size ratio is the set $B_1 = \{a_5, a_4\}$, with $\Gamma(B_1) = \{b_3\}$.
By Corollary~\ref{thm:lambda1}, the fairness parameter of the graph in the picture is $\frac 12$. But in order to actually match $a_5$ and $a_4$ with probability $\frac 12$, $b_3$ must be matched to one of the
two every single time. Hence the edge $(a_3, b_3)$ can never be used to in a maxmin-maxmin-fair solution. After removing $B_1$ and $\Gamma(B_1)$ from the graph,
 the next set of left vertices with smallest neighbor-to-size ratio is the set $B_2 = \{a_1, a_2, a_3\}$ and again we find that edge $(a_0, b_1)$ cannot be used. The last set we
 find in this way is $B_3 = \{a_0\}$.

 We refer to $B_1, B_2, B_3$ as the \emph{blocks} of the fair decomposition;  and to the increasing sequence of sets
$S_1 = B_1, S_2 = B_1 \cup B_2$ and $S_3 = B_1 \cup B_2 \cup B_3$ as the \emph{fairly isolated sets}. This motivates the following definitions.

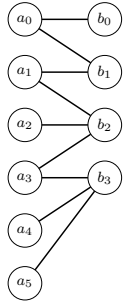
\begin{figure}[h]
  \begin{center}
\scalebox{0.7}{
\begin{tikzpicture}[scale=0.5]
\begin{scope}[rotate=-90]
    \grEmptyPath[Math,prefix=a,RA=2,RS=0]{6}
    \grEmptyPath[Math,prefix=b,RA=2,RS=3]{4}
\end{scope}
\Edges(b0,a0,b1,a1,b2,a3,b3,a5)
\Edges(a2,b2)
\Edges(a4,b3)
\end{tikzpicture}
}
  \caption{A bipartite graph with blocks $B_1 = \{a_5, a_4\}$, $B_2=\{a_3, a_2, a_1\}$ and $B_3=\{a_0\}$ and fairly isolated sets $S_1 = B_1, S_2 = B_1 \cup B_2$ and $S_3 = B_1
      \cup B_2 \cup B_3$.}\label{fig:hugegraph}
\end{center}
\vspace{-1mm}
\end{figure}

    For $A\subseteq L$, denote by $G\rst{A}$ the subgraph of $G$ induced by $A \cup \Gamma(A)$,
    and by $G/A$ the subgraph of $G$ induced by $(L \cup R) \setminus (A \cup \Gamma(A))$. 
    Intuitively, $G\rst{A}$ represents the subproblem where only the
    elements of $A$ are important, and $G/A$ represents the subproblem of $G\rst{\overline A}$ where 
the use of neighbours of $A$ is disallowed.
    For any subgraph $H$ of $G$, let $\pi(H)$ (resp., $\Pi(H)$) be the minimum (resp., maximum)
    satisfaction probability of an element of $V(H) \cap L$ in a maxmin-fair distribution.
The nonempty set $X\subseteq L$ is \emph{fairly isolated} if $\Pi(G\rst{X}) < \pi(G/X)$ or $X = L$.
This means that
every $u\notin X$ has guaranteed satisfaction larger than the largest maxmin-fair
satisfaction inside $X$,
even
if we remove
all  possibly conflicting edges from $X$ to~$\Gamma(X)$.

Finding fairly isolated sets enables a ``divide and conquer'' strategy to find maxmin-fair
distributions, since it turns out that 
matchings used inside $X$ have no bearing on the satisfactions needed for users
in $L\setminus X$.
For example, if we can determine that the set $B_1 \cup B_2$ is fairly isolated, then we can work independently on
$B_1 \cup B_2$ and $B_3$ and combine the distributions found.
\mycomment{
    \begin{restatable}{lemma}{lemcombine}
    \label{lem:combine}
    Let $X$ be fairly isolated. Then
    \begin{enumerate}[(a)]
        \item If $D_1$ is maxmin-fair for $G\rst{X}$ and $D_2$ is maxmin-fair for $G/X$, then $(S_1 \cup
            S_2)_{S_1\sim D_1, S_2 \sim D_2}$ is maxmin-fair for $G$.
        \item If $D$ is maxmin-fair for $G$, then the distributions of $(S \cap X)_{S \sim D}$ and $(S
                \cap \overline{X})_{S \sim D}$ are maxmin-fair for $G\rst{X}$ and $G/X$, respectively.
    \end{enumerate}
    \end{restatable}
}

With this in mind, we are ready
to state our fair decomposition theorem, proved in Appendix~\ref{sec:proofs_problem}:
\begin{theorem}\label{thm:decomp}
The fairly isolated sets form a chain $S_1 \subseteq S_2 \ldots
\subseteq S_{k-1} \subseteq S_k = L$. 
Define $S_0=\emptyset$ for convenience and 
let $B_i = S_i \setminus S_{i-1}$ for $i > 0$.
The following hold for all $i = 1,\ldots, k$:
\begin{enumerate}[(a)]
    \item $B_i \text{ is the maximal set } X \subseteq L\setminus S_i \text{ minimizing } \frac{ |\Gamma(X \cup S_i)| - |\Gamma(S_{i})| }{|X|}.$
    \item If $i < k$, $B_i \text{ is the maximal set } X \subseteq S_{i+1} \text{ maximizing } \frac{ |\Gamma(S_{i+1})| - |\Gamma(S_{i+1}\setminus X)| }{|X|}$.
    \item The satisfaction probability of every $v \in B_i$ in any maxmin-fair distribution is
    $\lambda_i = \frac{|\Gamma(B_i) \setminus \Gamma(S_{i-1})|}{|B_i|}$,
    and any $w \in \Gamma(B_i)\setminus \Gamma(S_{i-1})$ is matched to some $u \in B_i$ with probability 1. 
\end{enumerate}
We call $B_1, \ldots, B_k$ the \emph{blocks} in the fair decomposition of~$G$.
\end{theorem}



\subsection{Description of the basic algorithm}\label{sec:algorithms}
\label{sec:basic}
Theorem~\ref{thm:decomp} and Lemma~\ref{lem:lp_bipartite}
suggest a line of attack to solve the one-sided fair  bipartite matching problem, outlined in Algorithm~\ref{alg:outline} below.
\mycomment{
\begin{enumerate}[1.]
\item Find the blocks $B_0, \ldots, B_k$ in a fair decomposition.
\item For each 
{$i=0,1,2,\ldots, k$}, obtain a distribution $D_i$ of matchings
in $H_i = G/(\cup_{j\le i} B_j)$
such that $\cov{D_i}{u} = |\Gamma_{H_i}(B_i)|/|B_i|\; \forall u \in B_i$.
\item
Combine the distributions
$D_0, \ldots, D_k$ into a single maxmin-fair distribution $D$, and draw a matching from it.
\end{enumerate}
}
First, find the blocks $B_1, \ldots, B_k$ in a fair decomposition. Second, find a maxmin-fair distribution%
\footnote{Each of these distributions can be represented by a list of pairs $(probability, matching)$, with the
        probabilities being non-negative and summing up to 1. For our second algorithm a simpler representation is possible: the distribution  of matchings $D_i$ within each block (but
                not for the whole graph) is uniform over some small multiset of matchings.}
$D_i$
for each block $B_i$, using only edges that do not ``cross to  neighbors of lower blocks'' (i.e., no edge is
        allowed from $u  \in B_i$ to $v \in \Gamma( B_j )$ where $j < i$). 
        Finally,
combine the distributions
into a single maxmin-fair distribution $D$, and draw a matching from it.
Both our algorithms follow this general outline; they differ in how to perform steps 1 and 2. 
(We will discuss later (Section~\ref{sec:transparency}) an alternative implementation of step 3 which leads to distributions over a smaller number of matchings.)

Next we give the details of our first algorithm (Algorithm~\ref{alg:alg1}).

\begin{algorithm}
\caption{Outline of our polynomial-time algorithms for maxmin-fair matching}
\label{alg:outline}

\DontPrintSemicolon
\SetKwInput{Input}{Input}
\SetKwInput{Output}{Output}
\SetKwProg{Fn}{Function}{}{}

  \SetKwFunction{FMain}{MaxminFairMatching}
  \SetKwFunction{FD}{Test}
  \SetKwFunction{FBlock}{FairDecomposition}
  \SetKwFunction{FProb}{SingleBlockDistribution}
  \SetKwFunction{FCombine}{PickMatching}


 \Input{Bipartite graph $G=(V,E)$ with bipartition $V = L\ \dot{\cup}\ R$ and with $|\rho(L)| = R$}
 \Output{A maximum matching in $G$ drawn from a maxmin-fair distribution}\BlankLine

  \Fn{\FMain{$G$, $L$, $R$}}{
      \tcc{Step 1: find a fair decomposition}
      $ B_1,\ldots, B_k$ = \FBlock{$G$, $L$, $R$}
      \BlankLine

      \tcc{Step 2: obtain a fair decomposition for each block}
        \For{$i=1,\ldots, k$}{
            $R_i = \Gamma(B_i) \setminus \bigcup_{j < i} \Gamma(B_j) $\;
            $G_i = $ subgraph of $G$ induced by $B_i$ and $R_i$\;
            ${\calD}_i = $ \FProb{$G_i$, $L_i$, $R_i$}\;
        }
        \BlankLine

     \tcc{Step 3: combine the distributions and pick a matching}
     \For{$i=1,\ldots, k$}{
         $M_i = $ draw a matching from ${\calD}_i$
     }
        \KwRet{$\bigcup_{i=1}^k M_i$}\;
     
  }
\BlankLine

\end{algorithm}

\spara{\emph{Step 1: Finding a fair decomposition.}}
We will find the blocks  in a bottom-up manner.
To find the first block, observe the following:

\begin{lemma}\label{lem:gamma_union}
The maximal set minimizing $|\Gamma(X)|/|X|$ is 
the union of all non-empty sets $X$ minimizing $|\Gamma(X)|/|X|$.
\end{lemma}
\begin{proof}
It suffices to show that
if $X, Y$ are non-empty sets
minimizing $|\Gamma(X)|/|X|$, then $X \cup Y$ also
minimizes $|\Gamma(X)|/|X|$.
Indeed,
suppose $\frac{ |\Gamma(Y)| }{ |Y| } = \frac{ |\Gamma(X)| }{ |X| } \triangleq \lambda$.
By the submodularity of the cardinality of the neighborhood function of a graph,
   $$ |\Gamma(X \cup  Y)| + |\Gamma(X \cap Y)| \le |\Gamma(X)| + |\Gamma(Y)| = \lambda (|X| + |Y|).$$
   Notice that $|\Gamma(X \cup  Y)| \ge \lambda |X \cup Y|$ and $ |\Gamma(X \cap Y)| \ge \lambda |X \cap Y| $ by definition. If any of these two inequalities were strict we would have
   the contradiction
$$ |\Gamma(X \cup  Y)| + |\Gamma(X \cap Y)|  > \lambda (|X \cup Y| + |X \cap Y|) = \lambda (|X| + |Y|).$$
Hence the inequalities are not strict, and $|\Gamma(X \cup Y)| = \lambda |X \cup  Y|$.
\end{proof}

Along with Theorem~\ref{thm:decomp}, this observation suggests the following method, described in the \textsf{FindBlocks} method of Algorithm~\ref{alg:alg1}.
By solving 
{the}
LP 
{in}~\eqref{lp:mu} and using
Lemma~\ref{lem:lp_bipartite},
   we obtain
a set $X$ minimizing $|\Gamma(S)| / |S|$.
Remove $X$ from the graph $G$ and repeat (if $G$ is non-empty); let $Y$ be the new
set obtained.
If $\Gamma(Y) / |Y| = \Gamma(X) / |X|$, then replace $X$ with $X' = X \cup Y$ and  repeat the process of finding a minimizer $Y$ via
LP~\eqref{lp:mu}; this strictly increases the size of $X$.
Eventually we will  obtain a $Y$ satisfying
$|\Gamma(Y)| / |Y| > \Gamma(X) /|X|$, at which point we know that $X$ is the maximal set minimizing $\Gamma(S) /
    |S|$, i.e., the first non-trivial block $B_1$ is $X$. Now remove $B_1$ and $\Gamma(B_1)$ from $G$ and repeat (if applicable) to obtain $B_2, \ldots, B_k$.

\begin{algorithm}
\caption{First polynomial-time algorithm for maxmin-fair matching}
\label{alg:alg1}

\DontPrintSemicolon
\SetKwInput{Input}{Input}
\SetKwInput{Output}{Output}
\SetKwProg{Fn}{Function}{}{}

 \Input{Bipartite graph $G=(V,E)$ with bipartition $V = L\ \dot{\cup}\ R$}\BlankLine

  \SetKwFunction{FMain}{Main}
  \SetKwFunction{FD}{Test}
  \SetKwFunction{FS}{SmallestRatioSet}
  \SetKwFunction{FBlock}{FairDecomposition}
  \SetKwFunction{FProb}{SingleBlockDistribution}

  \Fn{\FS{$G$, $L$, $R$}}{
        Solve LP~\eqref{lp:mu} for the subgraph of $G$ induced by $L$ and $R$  \;
        $S = \{ v \in L \mid y_v > 0 \}$ \;
        $\lambda = \sum_{v \in R} y_v$ \;
        \KwRet{$S$, $\lambda$}\;
  } \BlankLine

  \Fn{\FBlock{$G$, $L$, $R$}}{
        $k = 0$ \;
        $L', R' = L, R$ \;

        \While{ $L' \neq \emptyset$ }{
            $X, \lambda' =$ \FS{$G$, $L'$, $R'$} \;
            $L' = L' \setminus X$ \;

            \If{ $k = 0$ {\normalfont or} $\lambda' \neq \lambda$ }
            {
\tcc{Create new block, possibly incomplete}
                $k = k + 1$ \;
                $B_k, \lambda = X, \lambda'$\;
                $R' = R' \setminus \Gamma_G(B_k)$ \tcc{Remove the neighbors of the previous block}
            } \Else{
\tcc{Merge with an existing block}
                $B_k = B_k \cup X$\;
            }
        }
        \KwRet{$B_1, \ldots, B_k$}\;
  } \BlankLine

  \Fn{\FProb{$G, L, R$}}{
      $F = $ a set of $|L| - |R|$ new right vertices\;
      $N = \{ (i, j) \mid i \in L, j \in F \}$ \;
      Add the new vertices $F$ and new edges $N$ to $G$ to form $G'$

      \tcc{ LP to find assignment probabilities }
      Find non-negative values $x_{uv}$ such that $\sum_{j \in \Gamma_{G'}(i)} x_{ij} = 1$  for all $j \in L \cup R$. \;
      \tcc{ Birkhoff-von Neumann decomposition }
      Find a distribution $D$ of matchings using edge $(u, v)$ with probability $x_{uv}$. \;
      Remove from each matching in $D$ the incident to $F$ \;
      \KwRet{$D$}\;
  }\BlankLine
\end{algorithm}

\spara{\emph{Step 2: Obtaining a fair distribution for each block.}}
The idea of this step is first to calculate the assignment probabilities $x_{uv}$ for all $u \in L, v \in R$, i.e., the probability that $u$ is matched to $v$ in some
fixed maxmin-fair distribution $F$. As of yet these probabilities are unknown (and, unlike satisfaction probabilities, they need not be the same
        for all maxmin-fair distributions). However, we do know some conditions that they must satisfy because we know 
(from Theorem~\ref{thm:decomp}) 
    the satisfaction probabilities of the left vertices  in $F$, and all the
right vertices need to  be  matched with probability 1 under our assumption that $\rho(L) = |R|$.
These conditions may be expressed as linear constraints in $x_{uv}$, so we will find suitable values for $x_{uv}$ via a linear program. Finally we can turn these values into an actual distribution
of matchings via the Birkhoff-von Neumann decomposition. Details follow.

   Consider the graph
   $H_i = G/\bigcup_{j\le i} B_j$ obtained by removing all lower blocks and their neighbors.
To simplify notation, rename $L \cap V(H_i)$ and $R \cap V(H_i)$ to $L$ and $R$.
We have $|R| \le |L|$ and $\lambda = |R| / |L| \le 1$.
First we calculate the (as of yet unknown) probabilities $x_{ij}$
    ($i \in L, j \in R$)
that each edge $(i, j)$ is {saturated} (i.e., $i$ is matched to $j$)
    in some fixed maxmin-fair distribution.
Clearly $\sum_j x_{ij} = \lambda$ for each $i$ and $\sum_i x_{ij} = 1$ for each $j$. Let us
add a set $Z$ of $|L| - |R|$ fictitious vertices to $R$ and extend the domain of definition of $x_{ij}$
so as
to satisfy
$x_{ij} = 1/|L|$ for each $i \in L, j \in Z$. We obtain a bipartite graph $G'$ with $|L|$ vertices on
each side; let $\Gamma'$ denote its neighborhood function.
Then
$         \sum_{v \in \Gamma'(u)} x_{uv} = 1 \; \forall u \in L,\;
         \sum_{u \in \Gamma'(v)} x_{uv} = 1  \;        \forall v \in R\cup Z,\; \text{and }
         x_{uv} \ge 0\;  \forall u\in L, v \in R\cup Z.$
         \mycomment{
                        The following linear program is then feasible:
                        {\small{
                                \begin{equation}\label{lp:mu_dual}
                                \begin{array}{rrclcl}
                                & \displaystyle \sum_{v \in \Gamma'(u)} x_{uv} &=& 1 & \forall u \in L \\
                                & \displaystyle \sum_{u \in \Gamma(v)} x_{uv} &=& 1          & \forall v \in R\cup Z \\
                                & x_{uv} &\ge& 0 & \forall u\in L, v \in R\cup Z \\
                                \end{array}
                                \end{equation}
                                }}
                        By solving~\eqref{lp:mu_dual} we can find a feasible solution $\{x_{uv}\}$.
         }
We can find a solution $x_{uv}$ to these inequalities by solving a linear program.

By the following consequence of Birkhoff-von Neumann theorem on doubly stochastic matrices~\cite{birkhoff}
the quantities $x_{uv}$ thus obtained represent
    the edge saturation probabilities of an actual distribution of matchings in~$G'$:%
    \begin{restatable}{lemma}{lembirkhoff}
                \label{lem:birkhoff}
                Let $\{x_{uv}\}_{(u, v) \in E}$ be non-negative numbers s.t. $\sum_{v \in R}
                x_{uv} \le 1
\;                \forall u \in L$ and $\sum_{u \in L} x_{uv} \le
                1\; \forall v \in R$. Then a distribution over $|E| + 1$ 
                matchings such that $\Pr_{M \in \calM} [ (u, v) \in M ]
                = x_{uv}$  exists and may be found in polynomial time.
                \end{restatable}
We thus obtain
a distribution $D$ of matchings in $G'$ 
in which each edge $(u, v)$ is used with probability~$x_{uv}$.
If we pick each matching with its probability in $D$
and remove from it
the edges incident
to the ''fictitious`` elements in~$Z$, we obtain a distribution of matchings where each element $i$ of
$L$ is matched with probability $1-\sum_{j \in Z} x_{ij} = 1-(|Z|/|L|) = 1-(|L|-|R|)/|L| = \lambda$, as desired.

\spara{\emph{Step 3: Combining the distributions.}}
The last step requires combining the distributions $D_1, \ldots, D_k$, each defined for a block
$B_i$,
    into a single maxmin-fair
distribution for $G$. The simplest way is to
 draw $(M_1, \ldots, M_k)$ from the product distribution $D_1 \times D_2 \ldots \times D_k$ and return
$M_1 \cup M_2 \ldots \cup M_k$. (This is an easily samplable maxmin-fair
        distribution with
        potentially large support.)

Putting all together, we obtain the following.
\begin{theorem} Algorithm~\ref{alg:alg1} is a polynomial-time algorithm for the one-sided maxmin-fair matching problem.  \end{theorem}





\section{A more efficient algorithm}
\label{sec:proofs_improved_algo}
    The algorithm from Section~\ref{sec:problem}
        requires
    solving   polynomially many LP
    subproblems. It was presented to showcase the main steps
    required, to introduce the fair decompositions, and to establish the existence of a polynomial-time algorithm.
In this section
we analyze a more efficient algorithm. It 
also follows each of the three steps outlined in Algorithm~\ref{alg:outline}, but differs from Algorithm~\ref{alg:alg1}
in two key respects:
\begin{itemize}
    \item For step 1, it finds fairly separated sets in arbitrary order, rather than bottom-up.
          These sets can be found by maximum flow computations in a certain graph, and a single flow computation can be used to find many new
          blocks in the decomposition simultaneously. 
   \item
   For step 2, it uses a fast
   edge-coloring algorithm on a carefully constructed regular bipartite graph, allowing us to bypass the (comparatively slow)
Birkhoff-von Neumann decomposition (for which the best known algorithm from~\cite{perfect_nlogn} runs in $\omega(|V| |E|)$ time).
\end{itemize}

We present pseudocode for the improved algorithm (Algorithm~\ref{alg:alg2}) at the end of this section.
We establish the following:
\begin{restatable}{theorem}{thmfairbpm}
\label{thm:fair_bpm}
Algorithm~\ref{alg:alg2} solves the maxmin-fair one-sided bipartite matching problem in $O((|V|^2 + |E| |V|^{2/3})\cdot (\log |V|)^2)$ expected time.
\end{restatable}

\subsection{Improved step 1: Finding a fair decomposition}
Suppose we wish to separate $L$ into vertices with
satisfaction probability $< \lambda$ and vertices with satisfaction probability $\ge \lambda$, for some parameter
$\lambda \in (0, 1)$. To perform this check, construct the graph $G(\lambda)$ by adding to $G$ a source vertex $s$
connected to every $u \in L$ with an edge of capacity $\lambda$, and a sink vertex $t$ connected to
every $v \in R$ with an edge of capacity 1; all other edges have infinite capacity.

\begin{lemma}\label{lem:cut}
Let $\kappa$ be the value of a minimum $s-t$ cut in  $G(\lambda)$.
Then exactly one of the following cases holds: (a)    $\kappa = \lambda |L|$ and $\pi(G) \ge
\lambda$; or (b) $\kappa < \lambda |L|$ and there is a fairly-isolated subset $X\subsetneq L$ such
that $\Pi(G\rst{X}) < \lambda$.
\; We can determine which case occurs, and obtain $X$ in case (b), with a min-cut computation on
$G(\lambda)$.
\end{lemma}
Recall that $\pi(G)$ (resp., $\Pi(G)$) represents the minimum (resp., maximum) satisfaction probabilities in a maxmin-fair distribution for $G$.
In either of the two cases contemplated by Lemma~\ref{lem:cut} we have ``made progress'' by solving a min-cut problem on $G(\lambda)$; either (a) we showed that achieving minimum satisfaction probability
$\lambda$ is possible, or (b) found a fair separation (and a reason why it is not possible).
\fullversion
{
\begin{proof}
Consider a minimum-value $s-t$ cut in $G(\lambda)$. Because the capacities of the edges from $s$ are no
larger than any other capacity, there  is always a cut $C$ with no larger value containing no edges from
$L$ to $R$. $C$ only contains edges from $s$ to some subset $\overline{A}_L \subseteq L$ and from
some subset $A_R \in R$ to
$t$;
 its value is
$\lambda |\overline{A}_L| + |A_R|$.

Let $A_L = L \setminus \overline A_L$ and $\overline{A}_R = R \setminus A_R$.
Because $C$ is an $s,t$-cut, there are no edges in $G$ (or in $G(\lambda)$) between $A_L$
and $\overline{A}_R$, so $\Gamma(A_L) \subseteq A_R$.
As $C$ is a \emph{minimum} cut, we must in fact have
$\Gamma(A_L) = A_R$ (or else cutting some edges from $A_R$ to $t$ is unnecessary). The value of $C$ is
$\lambda   |\overline{A}_L | + |A_R|$.
Furthermore, for any $X \subseteq \overline{A}_L$ we must have
$|\Gamma(X) \setminus \Gamma(A_L)| \ge \lambda |X|$, for
otherwise there would be a cut of smaller value
$$\lambda   |\overline{A}_L \setminus X| + |\Gamma(A_L \cup X)|
  = (\lambda |\overline{A}_L| + |A_R|) - \lambda |X| + |\Gamma(X) \setminus A_R|
.$$
So  the fairness parameter $\pi(G/{A_L})$ is at least $\lambda$.
If $A_L=\emptyset$, this is $\pi(G)$ and we are in case (a) of the Lemma.

If $A_L\neq\emptyset$, let $C=C(\lambda)$ be the minimum cut with minimum $|A_L|$.
Then $A_L$ is unique and may be determined in linear time by picking the vertices reachable from $s$ in the
 residual network of a maximum (pre)flow~\cite{all_mincuts}.
For any $Y \subseteq A_L$, we must have
$|\Gamma(A_L)| - |\Gamma(A_L \setminus Y)| < \lambda |Y|$, otherwise another cut $C'$ of at most the same same value but with $|A'_L| <
|A_L|$ would exist.
Hence $\Pi(G\rst{A_L}) < \lambda$ by Corollary~\ref{thm:lambda2} which, along with the previously derived inequality
$\pi(G/{A_L}) \ge \lambda$, states that $A_L$ is a fairly separated set, and we are in case (b).
\end{proof}
}

The parametric flow algorithm of Gallo et al~\cite{parametric_flow} can find the cuts $C(\lambda)$ in the proof of Lemma~\ref{lem:cut}
simultaneously for all $\lambda$ (in the sense of giving a cut for all possible $|L|-1$ ``breakpoints'' for
        $\lambda$). Its running time is asymptotically the
same time as that of a single maximum-flow computation via the push-relabel
algorithm of~\cite{push_relabel}.
However, this technique does not extend to all max-flow algorithms, and~\cite{push_relabel} is suboptimal for the graphs $G(\lambda)$.
A better idea is the following (see Algorithm~\ref{alg:alg2}).

Start with $\lambda = 1$ and keep halving
$\lambda$ as long as case (a) holds in Lemma~\ref{lem:cut}. The first time that (b) occurs we have found a fairly
separated  set $X$.
At this point we can  find recursively the blocks in the fair decompositions of $G\rst{X}$ and
$G/X$. The crucial insight is that \emph{we can find both in a single recursive call}:
$G\rst{X}$ and $G/X$ are disjoint, so min-cuts for
$(G\rst{X})(\lambda_1)$ and $(G/X)(\lambda_2)$ are easily obtained from min-cuts for a single
graph $G(\lambda_1, \lambda_2; X, \overline{X})$ containing a disjoint copy of each
(except that we still keep a single source $s$ and a single sink $t$).

An iterative implementation of this idea maintainins the following invariant: 
\begin{enumerate}[(a)]
\item we keep a partition of $L$ into $t \le k$ subsets $T_1, \ldots, T_t$;
\item each $T_i$ is
        the union of consecutive blocks in the decomposition 
        (in other words, it is the difference between two fairly isolated sets);
\item we have computed lower and
        upper bounds $\lambda_i$ and $\mu_i$ for the maxmin-fair probabilities of vertices in
        $T_i$, i.e.,
        $[\pi(T_i), \Pi(T_i)]\subseteq [\lambda_i, \mu_i)$;
\item these bounds satisfy
        $\mu_i -
        \lambda_i = 2^{-j}$ at iteration
        $j \ge 0$.

        \end{enumerate}
         Initially, $t =1$, $T_1 = L$, $\lambda_1 = 0$, $\mu_1 = 1$  (valid by the assumption $\rho(L) < |L|$), and $j = 0$.
         Construct the graph $G(\lambda_1', \ldots, \lambda_t'; T_1, \ldots, T_t)$ where
        the edge capacities from $s$ to each $u\in T_i$ are  $\lambda_i' = (\lambda_i + \mu_i) / 2$, and edges
        from $u \in T_i$ to $v \in \Gamma(T_j)$ where $j \neq i$ are deleted. With a
        min-cut computation in  $G(\lambda_1', \ldots, \lambda_t'; T_1, \ldots, T_t)$ we
        reduce the range of parameter bounds within $T_i$ by half for each $i$, and
        possibly split $T_i$ into two (increasing $t$) if we found a new fairly separated set.
        After the min-cut computation,
        obtaining the new partition of $L$, the new upper bounds, and removing the edges from lower
        blocks to higher ones takes linear time.

        After $O(\log |L|)$ iterations  (each performing a min-cut and a
                linear-time update), we have $\mu_i - \lambda_i <
        1/|L|^2$ for all $i$, at which point we have determined the full decomposition
        (because each maxmin-fair satisfaction probability is of the form
        $a/b$ where $a \le b, 1 \le b \le |L|$).
The running time of the max-flow algorithm of~\cite{flow_unit} for bipartite
                networks with rational capacities with denominators bounded by a polynomial in $|V|$ is
                $O(\min(|E|^{3/2}, |E| |V|^{2/3})\cdot \log |V|)$.
     We obtain:

\begin{theorem}\label{thm:decomp_time}
The fair decomposition of a graph $G=(V, E)$ for the one-sided fair bipartite matching problem
can be found in time $O(\min(|E|^{3/2}, |E| |V|^{2/3})\cdot (\log |V|)^2)$.
\end{theorem}

\begin{algorithm}
\caption{Improved (faster) polynomial-time algorithm for maxmin-fair matching}
\label{alg:alg2}

\DontPrintSemicolon
\SetKwInput{Input}{Input}
\SetKwInput{Output}{Output}
\SetKwProg{Fn}{Function}{}{}

 \Input{Bipartite graph $G=(V,E)$ with bipartition $V = L\ \dot{\cup}\ R$}\BlankLine

  \SetKwFunction{FBlock}{FairDecomposition}
  \SetKwFunction{FProb}{SingleBlockDistribution}
  \SetKwRepeat{Do}{do}{while}

  \Fn{\FBlock{$G$, $L$, $R$}}{
        $t = 1$ \tcc{number of sets} \;
        $T_1 = L$ \;
        $\lambda_1, \mu_1 = 0, 1$ \;
        \While{ $\mu_1 - \lambda_1 \ge \frac{1}{|L|^2}$}
        {
            Construct the graph $G' = G\big( \frac{\lambda_1 + \mu_1}{2}, \ldots, \frac{\lambda_t + \mu_t}{2}; T_1, \ldots, T_t\big)$ 
as in the discussion preceding Theorem~\ref{thm:decomp_time} 
                \;
            Run a max flow algorithm on $G'$ \;
            $X =$ set of vertices reachable from $s$ in the residual flow \;
            $p = t$ \;
            \For{$i = 1, \ldots, p$}{
                $X_i = X \cap T_i$ \;
                \If{ $X_i \neq \emptyset$ }{
                    \tcc{Separation found; split $T_i$ into two}
                    $T_{t+1}, \lambda_{t+1}, \mu_{t+1} = T_i \setminus X_i, \frac{\lambda_i + \mu_i}{2}, \mu_i$ \;
                    $T_i, \lambda_i, \mu_i = X_i, \lambda_i, \frac{\lambda_i + \mu_i}{2}$ \;
                    Remove from $G$ the edges between $T_{t+1}$  and $\Gamma(T_{i})$ \;
                    $t = t + 1$
                } \Else {
                    \tcc{Separation not found; update lower bound on $\pi(T_i)$}
                    $\lambda_i = \frac{\lambda_i + \mu_i}{2}$\;
                }
            }            
        }

        \KwRet{$T_1, \ldots, T_t$}\;
  } \BlankLine

  \Fn{\FProb{$G, L, R$}}{
$g=\gcd(|L|, |R|)$\;
$l = |L|/g$\; $r = |R| /g$\;
\BlankLine
Construct the graph $H = G(\lambda)$ as in Lemma~\ref{lem:cut} \;
    Find a maximum flow in $H$; let $x_{uv}$  denote the flow between $x \in L$ and $v \in R$ \;
\BlankLine
    Construct a multigraph $P$ with $x_{uv} \cdot r$ edges between each pair $(u, v) \in L \times R$ \;
      $F = $ a set of $|L| - |R|$ new right vertices\;
    Add to the right side of $P$ the vertices in $F$ \;
      $N = \{ (i, j) \mid i \in L, j \in F, i \equiv j \pmod r \}$ \;
      Add the edges in $N$ to $P$ \;
\BlankLine
          $C_1, \ldots, C_l = $ color classes in an $l$-coloring of the edges of $P$ \;
          Remove from $C_1, \ldots, C_l$ the edges incident to $F$ \;
          $D = $ the uniform distribution over $C_1, \ldots, C_l$  \;
    \KwRet{$D$}
  }\BlankLine
\end{algorithm}

\mycomment{
How to find a critical subset? Parametric max flow works.

Suppose we have a cut in a bipartite network. Without loss of generality the cut is always
src-something or something-sink. Say src to $\bar{A}_L$ and $A_R$ to sink. Then there are no edges
from $A_L$ to $\bar{A}_R$, so $\Gamma(A_L) \subseteq A_R$, so
$\bar{A}_L, A_R$ form a vertex
cover.
In fact we must have $\Gamma(A_L) = A_R$
because otherwise the vertex cover would not be minimum.
(The value is $\lambda |{\bar A}_L| + |A_R|$. )

Now for any $X \subseteq \bar{A_L}$ we have $\Gamma(X) \setminus A_R \ge \lambda |X|$, otherwise the cut would not be minimum.
    This means that $\lambda_1(\bar{A}_L)$ conditioned on $A_L$ is at least $\lambda$.

If in addition we have that $A_L \cup A_R$ is the unique minimal min-cut, then $A_L$ is minimal
itself, because $A_R = \Gamma(A_L)$ and removing stuff from $A_L$ will remove it from $A_R$ too.
So  $A_L$ is minimal ($\bar{A}_L$ is maximal). Then for
    for any $X \subseteq A_L$ we have $\Gamma(A_L)  \Gamma(A_L \setminus X) < \lambda |X|$,
    otherwise the cut would not be minimum in size. This means that $\lambda_n(A_L)$ is less than
    $\lambda$.

In short, the vertices reachable from $s$ (distances from $s$ less than $n$) belong to a block less than $\lambda$, and the rest have
fairness
$\ge \lambda$. Minimal cut on the $s$ side, and most cut edges on the $s$ side.
}

\subsection{Improved step 2: Obtaining a fair distribution for each block}
Here we describe the procedure 
in Algorithm~\ref{alg:alg2} to find fair distributions once the fair decomposition has been computed.
As before, suppose that $G$ itself has a single block, so $\Gamma(L) = \lambda |R|$.
Let $g=\gcd(|L|, |R|)$ and
$l = |L|/g$, $r = |R| /g$.

Let $G(\lambda)$ be as in Lemma~\ref{lem:cut}. By
the max-flow/min-cut theorem, there is a flow in $G(\lambda)$ with of value $\lambda |L| = |R|$.
Since
the incoming edges  to any $u \in L$ from $s$ have capacity $\lambda$, the flow from $s$ to $v$ must be
precisely $\lambda$. Let $x_{uv}$ be flow between $u \in L$ and $v \in R$. Then $\sum_{u \in L}
x_{uv} = \lambda = \frac{l}{r}$ and $\sum_{v \in R} x_{uv} = 1$, so we found  the edge saturation
probabilities  $\{x_{uv}\}$ of a maxmin-fair distribution. 

Consider now the subgraph $G'$ of $G$ containing only those edges for which $x_{uv} > 0$.
By Lemma~\ref{lem:birkhoff}, the same edge probabilities $x_{uv}$ warrant the existence of a distribution
of matchings in $G'$ with satisfaction probability $\lambda$.

By the integral flow theorem~\cite{lawler_book}, 
each $x_{uv}$ may be assumed to be a multiple of $1/r$, because all capacities in $G'$ are multiples of
$1/r$; in fact any standard maximum-flow algorithm
returns such a solution.
Now consider the $(r, l)$-biregular multigraph $P$ obtained by putting $n_{uv} = x_{uv} \cdot r$ parallel edges between $u \in L$
and $v \in R$. 
As in step 2 of Sec.~\ref{sec:basic}, we add to the right side of $P$
a set $Z$ of $|L \setminus R|$ vertices. 
Joining the $i$th vertex of $L$ with the  $j$th vertex of $Z$ whenever $i \equiv j \bmod g$, we
obtain from $P$ %
a graph $P'$ which is bipartite and $l$-regular.


    Any bipartite graph of maximum degree $l$ is $l$-edge-colorable so that no two adjacent edges share a color by 
   K\H{o}nig's theorem (see~\cite{matching_book}).
   Each color class is a matching, so there are $l$ 
    matchings in $P'$ covering each $u \in L$
    exactly $l$ times in total.
    \fullversion{Cole et al~\cite{edge_coloring} give an algorithm to color regular bipartite graphs
in time $O(m \log r) = O(m \log |\Gamma(L)|)$, where $m$ is the number of edges of $P'$;
in our case $m = O(l \cdot |E(G)|)$.}Goel et al~\cite{perfect_nlogn}
 give a randomized algorithm to color $l$-regular bipartite graphs
in expected time $O(l n^2 \log^2(n))$, where $n$ is the number of vertices of $P'$;
in our case $n = O(|L|)$ and \fullversion{we can use the crude bound} $l \le |L|$, so it runs in $O(|L|^2 \log^2(|L|))$.
If we remove the ``fictitious'' vertices in $Z$ from each of these matchings,
we are left with a multiset of $l$ matchings in $G$ covering each $u \in L$ exactly $r$ times. The
uniform distribution over them is thus maxmin-fair for $G$.

Now consider the case that the decomposition of $G$ has several blocks $B_1, \ldots, B_k$.
The values $x_{uv}^i$ for all blocks $i$ can be computed from a single maximum-flow computation in
$G(\lambda_1, \ldots, \lambda_k; B_1, \ldots, B_k)$ if
we know the blocks and each satisfaction probability $\lambda_i$. Then each corresponding coloring can be found in
time $O(n_i^2 \log^2(n_i))$; summing these running times and noticing that
$\sum_i n_i^2 \le |L|^2$,
       we deduce:
\mycomment{
Cole et al~\cite{edge_coloring} give an algorithm to color regular bipartite graphs
in time $O(m \log r) = O(m \log |\Gamma(L)|)$, where $m$ is the number of edges of $P'$;
in our case $m = O(l \cdot |E(G)|)$.
If we remove the ``fictitious'' vertices in $Z$ from each of these matchings,
we are left with a multiset of $l$ matchings in $G$ covering each $u \in L$ exactly $r$ times. The
uniform distribution over them is maxmin-fair for $G$.

Now consider the case that the decomposition of $G$ has several blocks $B_1, \ldots, B_k$.
The values $x_{uv}^i$ for all blocks $i$ can be computed from a single maximum-flow computation in
$G(\lambda_1, \ldots, \lambda_k; B_1, \ldots, B_k)$ if
we know the blocks and each satisfaction probability $\lambda_i$. Then each corresponding coloring can be found in
time $O(m_i \log |R|)$; summing these running times and noticing that
$\sum_i m_i \le |L| \cdot |E|$,
       we deduce:
}

\begin{theorem}
Given the fair decomposition, a maxmin-fair distribution for all blocks in it can be found in
$O(|V|^2\cdot (\log |V|)^2)$ expected time
\fullversion{
and $O(|V||E| \log |V|)$ deterministic time
}
after a max-flow computation.
\end{theorem}

\mycomment{
	Then the flow on each edge is a multiple of $1/a$,
    hence $G'$ . Extend it to an $a$-regular bipartite graph.

    Performing this for each block we get the needed matchings for each block.
    To reduce it to $n$...
}


Putting all together yields Theorem~\ref{thm:fair_bpm}.

\label{sec:transversal}

\section{Generalization to non-bipartite graphs}
    Recall that so far we have concerned ourselves with the {one-sided fair bipartite matching} problem, i.e., the special case of fair matching where
    $G$ is bipartite (with bipartition $V = L\ \dot{\cup}\ R$) and the set of users is $\calU=L$. 

    Notably, this special case can
    encode any other matching problem, and moreover we can make the simplifying assumption that $L$ is matchable and larger than $R$.
To show this, we make use of the following result from transversal theory.
\begin{theorem}[Ford and Fulkerson~\cite{transversals}]\label{thm:transversal}
For any graph $G$ there exists a bipartite graph $H$ with bipartition $(L_H, R_H)$ such that $L_H = V(G)$ and the collection of
matchable subsets of $V(G)$ in $G$ equals the collection of matchable subsets of $L_H$ in $H$.
\end{theorem}
\fullversion{This is normally stated as ``any matching matroid is transversal''.}
The construction of $H$ in Theorem~\ref{thm:transversal} can be carried out in polynomial time
(see~\cite{transversals_short} for a simple proof).
Hence the case of non-bipartite $G$
         can be reduced to the one-sided bipartite case.
A similar remark applies to general user sets $\calU \subseteq V$, as we can remove from $L_H$ the
elements of $V \setminus \calU$,
        which has no effect on the collection of matchable subsets of $\calU$ in $H$.

We make the additional simplifying assumption that $R$ is matchable.
If not,
find
an
arbitrary maximum matching of $G$ and remove from $R$ all unmatched vertices. Let $R'$ denote the
remaining vertices. 
\begin{theorem}[Mendelsohn and Dulmage~\cite{mendelson_dulmage}]If both $A \subseteq L$ and $B \subseteq R$ are
matchable in a bipartite graph $G$ with bipartition $(L, R)$, then
$A \cup B$ is also matchable in $G$.
\end{theorem}
It follows that for
each distribution of matchings of $G$ there is another distribution with the same coverage
(satisfaction)
probabilities for $L$ and covering only elements of $L \cup R'$.
Note that the
coverage probability  of each $v \in R'$ in this distribution is 1.
The case where $\rho(L) = |R|$ is easily handled separately (any maximum matching
        algorithm is maxmin-fair in this case), yielding the following (see Appendix~\ref{sec:proofs_problem} for details):
\begin{restatable}{theorem}{thmsimplify}
\label{thm:simplify}
The fair matching problem on arbitrary 
graphs with arbitrary user sets $\calU\subseteq V$ can be reduced
in polynomial time to the one-sided fair bipartite matching problem  on graphs where $\rho(L) = |R| < |L|$.
\end{restatable}
\begin{proof}
Let $\mathcal A$ be a maxmin-fair algorithm for one-sided bipartite matching problem described.
Given a graph $G$ and a user set $\calU$, we
\begin{enumerate}[(1)]
\item construct $H$ as in Theorem~\ref{thm:transversal};
\item remove from $L_H$ the elements of $V(G)\setminus \calU$;
\item find an arbitrary maximum matching $M$ and remove from $R_H$ the elements not covered
by $M$, using any polynomial-time maximum matching algorithm.
\item find a fair one-sided bipartite matching by either (a) using $\mathcal A$ on the resulting
graph if $|M| < |L_H|$ or (b) returning $M$ if $|M| = |L_H|$;
\item given the solution $S$ found at the previous step, return a matching in $G$ covering the same vertices
as $S$ using an explicit algorithm for Berge's theorem~\cite{edmonds_algo}.
\end{enumerate}
(Step (5) is technically redundant as we identify solutions with sets of matchable users, but
 is included for clarity.)

It is plain to see that the resulting distribution is maxmin-fair for the general problem if and only if $\mathcal A$
is maxmin-fair for the one-sided problem.
All steps run in polynomial time, possibly excluding the call to $\mathcal A$ itself.
\end{proof}

Combining Theorems~\ref{thm:fair_bpm} and~\ref{thm:simplify}, we obtain the following.
\begin{theorem}
The maxmin-fair matching problem on general unweighted graphs is solvable in polynomial time.
\end{theorem}

\section{On transparency and practical deployment}
\label{sec:transparency}
Even a provably fair algorithm might still be perceived by the
average user as a blackbox outputting an arbitrary solution. For the sake of transparency and
accountability, it can be interesting to publish all the solutions in a maxmin-fair distribution
    (along with their respective probabilities). Once a  complete fair distribution is published,
    convincing any user $u$ of fair treatment amounts to:
    \begin{enumerate}[(1)]
    \item letting $u$ verify independently the fairness
    guarantees of the distribution (for this it is also possible to output a short certificate, based on the {fair decomposition}, of the fact that no higher probability for $u$ is possible in a
    maxmin-fair distribution); and

    \item picking one of the published solutions at random, via any
    fair and transparent lottery mechanism or coin-tossing protocol (this is the only stage where
    randomness plays a role, as the distribution of matchings itself can be found deterministically).
    \end{enumerate}

One difficulty is the potentially large support size of the maxmin-fair distribution, which could
prevent publication.
An interesting question is if we can produce a maxmin-fair distribution with small support.
It turns out that for matchings,  $|L| - 1$ solutions always
   suffice; although  the actual number can be substantially smaller in practice (as shown in Section~\ref{sec:experiments}).
Let us discuss how to modify our algorithm so as to find a maxmin-fair distribution $F$ using at most $|L| + 1 - k$ matchings,
where $k$ is the number of non-trivial blocks in Theorem~\ref{thm:decomp}. 
(This could replace step 3 of Algorithm~\ref{alg:outline}.)
When $k = 1$, the technique from step 2 of Algorithm~\ref{alg:alg2} gives a multiset of $l \le |L|$
matchings. 

Consider the case $k=2$, which implies our claim for larger $k$ by induction.
Suppose $D$ (resp., $D'$) chooses matching $M_i, i \in [r]$ on $B_1$ (resp., $N_j, j \in [t]$ on $B_2$) with probability $p_i$ (resp., $q_j$).
    (Here $B_1 \cap B_2 = \emptyset$.)
    A simple greedy algorithm can construct a distribution $Z$ of matchings in $B_1 \cup B_2$ such
    that $\cov{D}{u} = \cov{Z}{u}$ for $u \in B_1$ and $\cov{D'}{u} = \cov{Z}{u}$ for $u \in B_2$ with at most $r + t
    - 1$ matchings, as follows.

    Keep indices $i \in [r]$, $j \in [t]$ and let $S$ denote a 
    set of
    (probability, matching) pairs, which will define the desired distribution at the end. At the
    outset $S=\emptyset$ and $i = j = 1$; at each iteration we
    add to $S$ the new pair $(\delta, M_i \cup N_j)$ where $\delta = \min(p_i, q_j)$.
    We decrement $p_i$ and $q_j$ by $\delta$ and increment $i$ (resp., $j$) if $p_i$ (resp, $q_j$)
    vanishes.
    The process terminates when $i$ and $j$ reach the end of their
    range, at which point $|S| = r + t - 1$ and all probabilities in $S$ sum up to 1.

    We note, however, that this procedure may produce some matchings with very small probabilities, so the precision needed to specify a maxmin-fair distribution exactly will grow.


\section{Experimental evaluation}
\label{sec:experiments}
We evaluate the practical performance of our fair matching algorithm by measuring
its running time and its ability to scale to large graphs, and analyzing the distribution
of maxmin-fair satisfaction probabilities and how they compare with those from two baselines. We also describe the features of the fair decompositions obtained.

\noindent\textbf{Reproducibility.}
Our code is available at \url{https://github.com/elhipercubo/maxmin_fair_bipartite_matching.git}.
It implements the improved algorithm from Sec.~\ref{sec:proofs_improved_algo}, with some implementation choices described below.
It was compiled with \texttt{g++} using \texttt{-O3} optimizations and run
on a dual-core  Intel i7-7560U CPU (2.40 GHz) with 16Gb RAM.

\spara{Datasets.}
We used publicly-available bipartite graphs of various types, sizes and domains:
all the graphs are already bipartite at the source repository, so that no preprocessing was needed.
Table~\ref{tab:datasets} reports their main characteristics.

\spara{Methods tested.}
We compare the following four methods to output maximum matchings:

\begin{description}
    \item[(\textsf{UF})] \emph{Unfair}: A standard maximum matching algorithm using maximum flows, optimized for runtime using the techniques from~\cite{hipr}. We use it for runtime comparisons only,
    because such a deterministic mechanism is inherently unfair as argued in Section~\ref{sec:intro}.
    \item[(\textsf{MF})] \emph{Maxmin fair}: Our mechanism, using the improved algorithm from Sec.~\ref{sec:proofs_improved_algo}.
    \item[(\textsf{PS})] \emph{Probabilistic Serial} (from~\cite{pserial}): The goal here is to find a set of edge flows from $L$ to $R$
    which can be  converted into a matching distribution by using the Birkhoff-von Neumann decomposition. 
    
    \textsf{PS} attempts to find a fair flow via a greedy algorithm,  as follows:
    each user $u \in L$ sends flow at the same fixed rate, sharing this rate equally among her
    neighbours. When the outgoing flow of $u \in L$ (or the ingoing flow of $v \in R$) reaches 1,
    remove $u$ (or $v$). Repeat while there are edges remaining.

    Unlike~\textsf{MF}, this mechanism is not Pareto-efficient (i.e., it does not necessarily return maximum matchings), but like~\textsf{MF}, a single run of the mechanism can be
    used to output all satisfaction probabilities.

     \item[(\textsf{RP})] \emph{Random Priority} (see~\cite{pserial}): It finds a matchable set of vertices as follows:
     
        Let $S = \emptyset$. Process all users in random order, adding user $u$ to
        $S$ if $S \cup \{u\}$ is matchable. Return a matching covering the final set $S$. 

     Like~\textsf{MF}, this mechanism is Pareto-efficient, but unlike~\textsf{MF}, a single run of the mechanism only outputs a single matching and hence cannot be
    used to compute all satisfaction probabilities.
\end{description}
The latter two methods arose from work in economics in a different setting:
 randomized assignments on full bipartite graphs with \emph{ordinal preferences}, i.e., where every $u \in L$ has a full ranking of
all possible partners $v \in R$, and the goal is to design mechanisms which are ordinally efficient.
(By contrast, in our setting the graph is not complete but there are no ordinal preferences: each user considers all of its neighbours equally desirable.) However they can be
naturally applied in our context as well.

\spara{Implementation.} We used the improved algorithm from Sec.~\ref{sec:proofs_improved_algo}.
\mycomment{
, except that (a)
   during step 1 we build the graph $G(\lambda_1', \ldots, \lambda_t'; S_1, \ldots, S_t)$ with
   $\lambda_i' = \lambda_i$ instead of $\lambda_i' = (\lambda_i + \mu_i) / 2$,
which works well despite invalidating
the theoretical runtime bound;
and (b) when $\lambda_i = \mu_i$ is reached for a set
$S_i$ we remove $S_i \cup \Gamma(S_i)$ from $G$ for speedup. Neither change affects
correctness.

}
For max flow computations we chose the highest-label push-relabel algorithm of~\cite{push_relabel},
which performs best 
with the gap heuristic
from~\cite{eff_pushrelabel}. We follow Goldberg techniques from~\cite{hipr}: 
efficient gap detection is done via bucket lists of active nodes at each
level\ifhideproofs~\cite{eff_pushrelabel}\else{~\cite{hipr}}\fi,
and we arrange edges from/to the same vertex consecutively to take advantage of cache locality.
     We avoid floating-point computations by using exact integral multipliers.

     For reasons of simplicity and/or practical efficiency,
our implementation departs from the pseudocode in Algorithm~\ref{alg:alg2} in the points below. None of these changes affect correctness.
\begin{itemize}
    \item In \textsc{FairDecomposition} (line 8), only a pre-flow algorithm (the first phase of~\cite{push_relabel}) is run. This always suffices to find min-cuts (not max-flows) and thus fairly
    separated sets, and can halve the runtime.

    \item In \textsc{FairDecomposition} (line 7), 
     instead of 
   setting $\lambda_i' = (\lambda_i + \mu_i) / 2$
when building the graph $G(\lambda_1', \ldots, \lambda_t'; T_1, \ldots, T_t)$,  
   we set it to $|\Gamma(T_i)|/|T_i|$. That is, our implementation guesses optimistically that $T_i$ is actually a single block in the
   decomposition, with the pre flow computation used to verify that guess or split the block in two. This also allows us to change the terminating condition (line 6) and stop earlier:
   rather than stopping when $\lambda_i$ and $\mu_i$ are very close, which may occur long after the
   full decomposition have been found, we stop when no block is split. This new choice for $\lambda_i'$ may invalidate our
   theoretical bound on the number of flow computations required, but it makes the code much faster in practice.
   \item In \textsc{SingleBlockDistribution} (lines 27--32), we do not always build $P$ because
   $P$ may be quite large for some blocks with a small number of right vertices (where $r \ll l$), due to the fictitious edges added. Rather, we first attempt to find $l$ disjoint matchings  of size $r$
   in arbitrary order without including fictitious vertices/edges.
   This often succeeds and, when it does, gives a correct coloring. When this fails, we build $P$ and proceed to find the coloring as described.
   \item In \textsc{SingleBlockDistribution} (line 32), we do not use Goel et al's edge coloring algorithm on $P$. Rather, we find $l$ matchings of size $l$ one by one. This simplifies the
   implementation, but may impact performance and  the runtime bound.
\end{itemize}


\begin{table}[h!]
\begin{center}
\begin{footnotesize}
\centering
\caption{Datasets used: code, number of left and right nodes ($|L|$, $|R|$), number of
    edges ($|E|$), maximum matching size ($\rho$). Available at: \url{http://konect.uni-koblenz.de/networks/}
        \label{tab:datasets}}

\begin{tabular}{c|ccccc}
\multicolumn{1}{c}{dataset (code)} &  \multicolumn{1}{c}{$|L|$} &
\multicolumn{1}{c}{$|R|$}          &  \multicolumn{1}{c}{$|E|$} & \multicolumn{1}{c}{$\rho$}         \\ \hline
\textsf{actor-movie (AM)}     & 127,823    & 383,640    & 1,470,404   & 114,762       \\ 
\textsf{pics-ti (Vui)}         & 82,035     & 495,402     & 2,298,816   & 67,608       \\ 
\textsf{citeulike-ti (Cti)}    & 153,277    & 731,769    & 2,411,819   & 120,125        \\ 
\textsf{bibsonomy-2ti (Bti)}   & 204,673    & 767,477    & 2,555,080   & 152,757    \\ 
\textsf{wiki-en-cat (WC)}     & 1,853,493  & 182,947    & 3,795,796   & 179,546 \\ 
\textsf{movielens (M3)}       & 69,878     & 10,677     & 10,000,054  & 10,544\\ 
\textsf{flickr (FG)}          & 395,979    & 103,631    & 8,545,307   & 96,866 \\ 
\textsf{dblp-author (Pa)}     & 1,425,813  & 4,000,150  & 8,649,016   & 1,425,803     \\ 
\textsf{discogs-aff (Di)}     & 1,754,823  & 270,771    & 14,414,659  &248,796  \\ 
\textsf{edit-dewiki (de)}     & 425,842    & 3,195,148  & 57,323,775  &  355,045\\ 
\textsf{livejournal (LG)}     & 3,201,203  & 7,489,073  & 112,307,385 &2,171,971 \\ 
\textsf{trackers (WT)}        & 27,665,730 & 12,756,244 & 140,613,762 & 4,006,867 \\ 
\textsf{orkut (OG)}           & 2,783,196  & 8,730,857  & 327,037,487 & 1,980,077 \\ \hline
\end{tabular}
\end{footnotesize}
\end{center}
\end{table}


\begin{table}[h]
\begin{center}
\centering
\caption{Characteristics of fair decompositions: number of blocks ($k$), edges used ($e_1$), number of matchings in the fair distribution ($M$).
    \label{tab:decomp}}


\begin{tabular}{c|ccc}
\multicolumn{1}{c}{dataset}   & \multicolumn{1}{c}{$k$}  &
\multicolumn{1}{c}{$e_1$}         &
 \multicolumn{1}{c}{$M$}
\\ \hline
\textsf{AM}  & 194  & 143,425      & 13,762
\\ 
\textsf{Vui} & 72   & 84,003       & 1,068
\\ 
\textsf{Cti} & 151  & 157,744      & 2,726
\\ 
\textsf{B}   & 179  & 212,667      & 4,119
\\ 
\textsf{WC}  & 1468 & 1,883,431    & 350,518
\\ 
\textsf{M3}  & 245  & 92,841       & 52,332
\\ 
\textsf{FG}  & 924  & 435,612      & 109,242
\\ 
\textsf{Pa}  & 2    & 1,425,813    & 2
\\ 
\textsf{Di}  & 1117 & 1,784,259    & 305,104
\\ 
\textsf{de}  & 163  & 432,472      & 5,596
\\ 
\textsf{LG}  & 1480 & 3,314,628    & 302,410
\\ 
\textsf{WT}  & 3612 & 27,842,321   & 16,548,387
\\ 
\textsf{OG}  & 2266 & 3,041,112    & 224,738
\\ \hline
\end{tabular}
\end{center}
\end{table}

\begin{table}[h]
\begin{center}
\centering
\caption{Running time (in seconds) of \textsf{UF}, \textsf{MF}, \textsf{PS} and \textsf{RP}. 
    The ${}_1$ subscript refers to the time to compute assignment probabilities in the solution without converting them into a distribution of matchings (only meaningful for
            \textsf{MF} and 
            \textsf{PS}). Dashes indicate running times above one hour.
    \label{tab:all_times}}


\begin{tabular}{c|ccccccccc}
\multicolumn{1}{c}{dataset}
& \multicolumn{1}{c}{\textsf{UF}} 
& \multicolumn{1}{c}{$\textsf{MF}_1$}
& \multicolumn{1}{c}{\textsf{MF}}
& \multicolumn{1}{c}{$\textsf{PS}_1$}
& \multicolumn{1}{c}{\textsf{RP}} 
\\ \hline
\textsf{AM}
& 0.34
& 1.812 & 9.309  & 832 & 11.27
\\ 
\textsf{Vui}
& 0.38 
& 1.206 & 1.251 & 717 & 0.39
\\ 
\textsf{Cti}
& 0.42 
& 1.731 & 1.820 & 1443 & 0.59
\\ 
\textsf{B}
& 0.47 
& 2.152 & 2.293 & 1764 & 0.74
\\ 
\textsf{WC}
& 1.68 
& 17.57 & 23.036& - & 68.17
\\ 
\textsf{M3}
& 1.08 
& 11.12 & 319.13 & 485 & 72.18
\\ 
\textsf{FG}
& 1.16 
& 15.17 & 25.80 & - &688.21
\\ 
\textsf{Pa}
& 2.99 
& 5.96  & 6.908 & - & 3.20
\\ 
\textsf{Di}
& 3.10 
& 23.37 & 24.714 & -& 67.83
\\ 
\textsf{de}
& 6.99 
& 21.67 & 22.007 & -& 52.84
\\ 
\textsf{LG}
& 26.05 
& 103.59 & 108.311 & -& -
\\ 
\textsf{WT}
& 51.92
& 444.71 & 2270.81 & -& -
\\ 
\textsf{OG}
& 98.20
& 370.76 & 381.524 & -& -
\\ \hline
\end{tabular}
\end{center}
\end{table}

\subsection{Fair decompositions: characteristics}
Table~\ref{tab:decomp} shows  the number of blocks $k$ in the fair decomposition (for informative purposes), the number of
distinct edges $e_1$ used
in a maxmin-fair distribution $F$,
    the number of matchings
$M$ in the support of $F$, and the time  to compute the decomposition.
As we anticipated in Section \ref{sec:transparency} the number of matchings needed for a fair distribution ($M$) is in practice much smaller than $|L|$.
Another observation  is that $e_1$
   exceeds $|L|$ only slightly.  This is a measure of the storage needed to publish a summary containing the fair decomposition, the satisfaction probabilities, and the probability of each
   edge being used in the matching, which can be verified independently. (Publishing an explicit list of $M$ matchings of size $\rho$ explicitly
           would take much more space as many of these matchings share many edges.)


\subsection{Running time}
    In Table~\ref{tab:all_times} we present runtimes of all four methods 
    for the datasets considered. Dashes indicate times above one hour.
 We report user times; the real times are within $2\%$ of these in all cases except for~\textsf{OG}, where the memory needs for graph and data structures
 exceeded the RAM available (16Gb), causing excessive disk swapping.

    As $\textsf{MF}$ has two clearly differentiated parts, we  analyze two different runtimes:
\begin{itemize}
    \item Time to compute the satisfaction probabilities of each user, and the probability of using each edge in the maxmin-fair distribution (step 1, finding the fair
            decomposition), reported in column $\textsf{MF}_1$;
    \item Total time, including all of the above plus the time to find an edge coloring for each block and the list of matchings for each block (step 2), 
    reported in column $\textsf{MF}$. (The time to draw a matching given this list, step 3, is negligible.)
\end{itemize}    
No clear pattern emerges as to which of these two phases is faster in practice.
As can be observed, in some instances (\textsf{Vui}, \textsf{Cti}, \textsf{B}, \textsf{Di}, \textsf{de}, \textsf{LG}, \textsf{OG}) the time is dominated by the first phase, wheres in others 
(\textsf{AM}, \textsf{M3} and \textsf{WT}) the total time is much larger than the time for phase 1 only;
in the latter cases
the exact requirement of maxmin-fairness forces the algorithm to need a large number of
matchings for some blocks, increasing the time for step 2. It seems likely that a more relaxed requirement of approximate fairness could lead to vast improvements in the runtime of
step 2.

Similarly, for \textsf{PS} we differentiate between the time to compute the probability of each edge being used (column $\textsf{PS}_1$), and the total time including the former
plus the time to find
a
distribution of matchings which agrees with those probabilities.
However, this last step is much slower in the case of $\textsf{PS}$ because in this case we need to use the full Birkhoff-von Neumann decomposition, instead of exploiting the
degree regularity conditions of the blocks to find edge colorings as we do for $\textsf{MF}$. Because 
existing implementations of the Birkhoff-von Neumann decomposition do not scale even for the smaller graphs tested, we decided to omit the second phase of $\textsf{PS}$ (which is not required to analyze its fairness properties).

As is to be expected, the unfair algorithm~\textsf{UF} is the fastest. As for the others, the only one which can be run to completion within one hour in all datasets is
ours (\textsf{MF}). Its runtime is usually a handful of seconds except for the very large graphs, where it is in the order of minutes (37 minutes at most, attained for the graph
        \textsf{WT}).
    We can see that \textsf{PS} is the most computationally expensive, as many iterations of its
    main loop are required to reach convergence, and each
    iteration
    takes linear time.
    Finally, the runtime of \textsf{RP} is generally comparable to that
    of \textsf{MF} on small and medium-size graphs, outperforming it on many of the smaller graphs, but \textsf{RP} becomes slower for larger graphs; despite the additional complexity of \textsf{MF}, the priority mechanism of \textsf{RP} precludes the use of the
    push-relabel max-flow algorithms, and also limits the number of simultaneous augmenting paths
    which can be found during a single graph search in augmenting-path algorithms.  Notice that these are runtimes for a single run;
    if the
    satisfaction probabilities need to be computed then it becomes necessary to run \textsf{RP} a large number of times,
                 slowing it down considerably. (This is not the case for \textsf{MF} or \textsf{PS}.)

\subsection{Satisfaction probability comparison}
Next we analyze the satisfaction probabilities produced by maxmin-fair matchings (\textsf{MF}) and compare
with the probabilistic serial and random-priority mechanisms. Note that the exact determination of the satisfaction probabilities of \textsf{RP} is computationally infeasible.
        To approximate them, we run \textsf{RP} a total of $T = 1000$ times with
        independent uniformly random permutations. (Note, however, that this does not give a good estimate of probabilities below $1/T$.)

        For this comparison we focus on the smaller graphs,
\emph{due to the limited scalability of
    \textsf{PS} (which needs a large number of iterations in its main loop, each taking linear time)
    and
\textsf{RP} (which needs to be run $T$ times to approximate the satisfaction probabilities)}.

Finally, table~\ref{tab:fairness} reports the distribution of satisfaction probabilities:
minimum value ($\lambda_{\text{min}}$), quantiles, percentage of
users with satisfaction 1 ($per1$), and Nash welfare ($N_0$),
    the geometric mean of utilities (satisfaction probabilities in our setting). %
\mycomment{:
    \begin{equation}\label{eq:nash}
      N_0(D) = \left(\prod_{u \in \calU} \cov{D}{u} \right)^{1/|\calU|}.
    \end{equation}
Nash welfare is a standard measure of fairness when allocating divisible
    resources~\cite{nash_unreasonable}. 
    As in~\cite{gen_nash}, we also study the generalization of Nash welfare using
    power means (for a parameter $p \in \mathbb{R}$):
    \begin{equation}\label{eq:gen_nash}
      N_p(D) = \left( \frac{\sum_{u \in \calU} \cov{D}{u}^p}{|\calU|} \right)^{1/p}.
     \end{equation}
    }%
%
Nash welfare is a standard measure of fairness when allocating divisible
    resources~\cite{nash_unreasonable}. 
    As in~\cite{gen_nash}, we also study the generalization of Nash welfare using
    power means (for a parameter $p \in \mathbb{R}$):

    \mycomment{
    $$
      N_0(D) = \left(\prod_{u \in \calU} \cov{D}{u} \right)^{1/|\calU|}\, \eqnum\label{eq:nash}; \;
      N_p(D) = \left( \frac{\sum_{u \in \calU} \cov{D}{u}^p}{|\calU|} \right)^{1/p}\,\eqnum\label{eq:gen_nash}
.
    $$
    }

    \begin{equation}\label{eq:nash}
      N_0(D) = \left(\prod_{u \in \calU} \cov{D}{u} \right)^{1/|\calU|},
     \end{equation}
    \begin{equation}\label{eq:gen_nash}
      N_p(D) = \left( \frac{\sum_{u \in \calU} \cov{D}{u}^p}{|\calU|} \right)^{1/p}.
     \end{equation}

    \mycomment{:
    \begin{equation}\label{eq:nash}
      N_0(D) = \left(\prod_{u \in \calU} \cov{D}{u} \right)^{1/|\calU|}.
    \end{equation}
Nash welfare is a standard measure of fairness when allocating divisible
    resources~\cite{nash_unreasonable}. 
    As in~\cite{gen_nash}, we also study the generalization of Nash welfare using
    power means (for a parameter $p \in \mathbb{R}$):
    \begin{equation}\label{eq:gen_nash}
      N_p(D) = \left( \frac{\sum_{u \in \calU} \cov{D}{u}^p}{|\calU|} \right)^{1/p}.
     \end{equation}
    }
When $p = 1$, $N_p(D)$ is the mean satisfaction probability,
     which
    equals $\rho / |L|$ for any Pareto-efficient mechanism.
Taking the limit in~\eqref{eq:gen_nash} as $p \to 0$ one obtains~\eqref{eq:nash}~\cite{ineq_book}, justifying the notation $N_0$ for (standard) Nash welfare. Taking the limit as $p \to -\infty$ yields $\min_{u \in \calU} \cov{D}{u}$, which by definition is maximized by MF.

\begin{table}
\begin{center}
\begin{small}
\centering
\caption{Distribution of maxmin-fair satisfaction probabilities.
        \label{tab:fairness}}
\begin{tabular}{c|cccccccc}

\multicolumn{1}{c}{dataset}  & \multicolumn{1}{c}{$\lambda_{\text{min}}$} & \multicolumn{1}{c}{$\lambda_{25\%}$} & \multicolumn{1}{c}{$\lambda_{50\%}$} & \multicolumn{1}{c}{$\lambda_{75\%}$} &
   \multicolumn{1}{c}{$per1$}
  & \multicolumn{1}{c}{$N_0$}
\\ \hline
\textsf{AM}   & 0.156& 1 & 1 & 1 
                  & 76.18
                  & 0.860
                \\ 
\textsf{Vui}    & 0.0208 & 0.5 & 1 & 1 
               & 69.16
               & 0.743
                \\ 
\textsf{Cti}    & 0.01   & 0.5 & 1 & 1 
                 & 65.42
                 & 0.670
                \\ 
\textsf{B}    & 0.0128  & 0.5 & 1 & 1 
                & 58.63
                & 0.630
                \\ 
\textsf{WC}     & 5.52e-4 & 0.0149 & 0.0417    & 0.1 
               & 2.17
               & 0.0366
                \\ 
\textsf{M3}     & 0.0298 & 0.0833 & 0.137 & 0.1798 
               & 0.49
               & 0.121
                \\ 
\textsf{FG}     & 0.001012& 0.116 & 0.2 & 0.254 
                & 6.04
                & 0.161
                \\ 
\textsf{Pa}    & 0.5    & 1 & 1 & 1 
                  & 99.9999
                  & 0.99999
                \\ 
\textsf{Di}    & 0.000025 & 0.0135 & 0.0588 & 0.167 
                 & 4.28
                 & 0.0377
                \\ 
\textsf{de}   & 0.00334 &  0.667 & 1 & 1 
                 & 74
                 & 0.718
                \\ 
\textsf{LG}  & 1.83e-4  & 0.2 & 1 & 1 
                 & 59.43
                 & 0.326
                \\ 
\textsf{WT}    & 4.38e-7    & 4e-6 & 1.26e-4 & 0.0294 
                  & 10.89
                  & 0.000366 
                \\ 
\textsf{OG}     & 0.00453       &  0.333 & 1 & 1 
                & 56.2
                & 0.583
                \\ \hline
\end{tabular}
\end{small}
\end{center}
\vspace{-2mm}
\end{table}

\begin{table}
\begin{center}
\centering
\caption{Generalized Nash welfare of MF, PS and RP. Larger is better.\label{tab:gen_nash}}
\begin{tabular}{|c|ccccc|}
\hline
dataset             & \textsf{AM} & \textsf{Vui} & \textsf{Cti} & \textsf{B} & \textsf{M3} \\ \hline
\hline
$N_{ 1} (MF)$        &0.898&0.824&0.784&0.746&0.151 \\
$N_{ 1} (PS)$        &0.796&0.775&0.739&0.702&0.147 \\
$N_{ 1} (RP) $       &0.898&0.824&0.784&0.746&0.151 \\ \hline
\hline

$N_{ 0} (MF)$        &0.860&0.743&0.670&0.630&0.121 \\
$N_{ 0} (PS)$        &0.695&0.684&0.618&0.580&0.0824 \\
$N_{ 0} (RP)$        &0.855&0.739&0.667&0.622&0.117 \\ \hline
\hline

$N_{-1} (MF)$        &0.797&0.585&0.470&0.445&0.0927 \\
$N_{-1} (PS)$        &0.524&0.522&0.419&0.395&0.0452 \\
$N_{-1} (RP)$        &0.780&0.579&0.460&0.430&0.0868 \\ \hline
\hline

$N_{-2} (MF)$        &0.699&0.339&0.232&0.244&0.0720 \\
$N_{-2} (PS)$        &0.335&0.300&0.210&0.214&0.0285 \\
$N_{-2} (RP)$        &0.659&0.332&0.228&0.230&0.0650 \\ \hline
\hline

$N_{-5} (MF)$        &0.409&0.0822&0.0431&0.0573&0.0486 \\
$N_{-5} (PS)$        &0.100&0.0767&0.0429&0.025&0.0165 \\
$N_{-5} (RP)$        &0.333&0.0793&0.0429&0.0553&0.0397 \\ \hline
\end{tabular}
\end{center}
\end{table}

Table~\ref{tab:gen_nash} shows these metrics
    for the three mechanisms tested, on those graphs where PS terminated in 8 hours. 
Notice that $N_1(\textsf{MF}) = N_1(\textsf{RP}) > N_1(\textsf{PS})$, confirming that \textsf{MF} and \textsf{RP} are Pareto-efficient but \textsf{PS} is not.
The generalized welfares for $p < 1$ are computed exactly for \textsf{MF} and \textsf{PS}, but
    estimated from the empirical probabilities after $T$ samples for \textsf{RP}. %
({For $p=0$ we replace each empirical probability $q$ by $\max(q, 1/T)$ so that the estimate is non-zero.}) 
    MF comes out on top for all (generalized) Nash welfares in all instances, in accordance with a result of~\cite{dichotomous}.
    Interestingly, $N_0(\textsf{RP})$ is typically within 1\% of $N_0(\textsf{MF})$ (as both solutions result in a large
    proportion of users with high satisfaction), but for smaller $p$ the gap
    can widen to as much as $22\%$ for $p = -5$, in accordance with the fact
    that MF was designed to provide better guarantees to low-satisfaction users.

Table~\ref{tab:accum} shows the expected fraction of satisfied users among the bottom
$t\%$ using each method, for $t = 1, 5, 10$, and $20$.
Again, our method always gives the highest values, the gap with the second best being as high as 50\% in
some instances where $t=1$.

Table~\ref{tab:variance} reports the variance
of the logarithms of satisfaction probabilities over uniformly random users,
as a measure of inequality. (Taking logs penalizes wildly varying ratios of satisfaction
        probabilities.)
We see that \textsf{MF}, which minimizes social inequality in the sense of Def.~\ref{def:social},
also tends to minimize this quantity in all datasets tested.

\begin{table}
\begin{center}
\begin{small}
\centering
\caption{Fraction of satisfied users among the bottom $t\%$.\label{tab:accum}}
\begin{tabular}{|c|ccccc|}
\hline
dataset             & \textsf{AM} & \textsf{Vui} & \textsf{Cti} & \textsf{B} & \textsf{M3} \\ \hline
\hline
$MF, t=1$ &0.189&0.0610&0.0481&0.0459&0.0298 \\
$PS, t=1$ &0.0640&0.0528&0.0417&0.039&0.00430 \\
$RP, t=1$ &0.147&0.0518&0.0456&0.0439&0.0198 \\\hline
\hline

\hline
$MF, t=5$ &0.282&0.158&0.100&0.0970&0.0298 \\
$PS, t=5$ &0.118&0.130&0.0870&0.0825&0.00605 \\
$RP, t=5$ &0.256&0.144&0.0045&0.0914&0.0252 \\\hline
\hline

\hline
$MF, t=10$ &0.389&0.231&0.151&0.143&0.0338 \\
$PS, t=10$ &0.167&0.192&0.128&0.123&0.00883 \\
$RP, t=10$ &0.363&0.216&0.139&0.133&0.0290\\\hline
\hline

\hline
$MF, t=20$ &0.513&0.344&0.238&0.221&0.0381 \\
$PS, t=20$ &0.259&0.287&0.204&0.193&0.0157 \\
$RP, t=20$ &0.506&0.326&0.229&0.212&0.0353 \\\hline
\end{tabular}
\end{small}
\end{center}
\end{table}

\begin{table}
\begin{center}
\centering
\caption{Inequality measure: variance  of log-satisfaction probabilities. Smaller is better.\label{tab:variance}}
\begin{tabular}{|c|ccccc|}
\hline
dataset             & \textsf{AM} & \textsf{Vui} & \textsf{Cti} & \textsf{B} & \textsf{M3} \\ \hline
\hline

$\Var[\log (MF)]$    &0.112&0.296&0.454&0.475&0.491 \\
$\Var[\log (PS)]$    &0.385&0.349&0.518&0.534&1.391 \\
$\Var[\log (RP)]$    &0.133&0.304&0.475&0.496&0.556 \\ \hline
\end{tabular}
\end{center}
\end{table}

\section{Related work}
\label{sec:related}

To the best of our knowledge, we are the first to study computationally efficient randomized maxmin-fair matching algorithms, and to offer a general definition of fairness for general search
problems.

The work of Bogomolnaia and Mouline~\cite{dichotomous} on random matching under dichotomous preferences is closely related to ours:
they define an \emph{egalitarian solution} 
and show that it is envy-free, strategy-proof and group-strategy-proof with respect the set of right or left vertices.
As the authors note, they do not provide an axiomatic characterization of their solution; rather, their definition of egalitarian is expressed in terms of
a specific algorithm and is thus not easily generalizable to other search problems. By contrast, our definition of distributional maxmin-fairness
applies to any search problem with non-unique solutions and, in the special case of bipartite matchings, is equivalent to the egalitarian solution.
In~\cite{dichotomous} two simple algorithms are proposed
to find egalitarian matchings, both of them running in exponential time; our work yields a practical polynomial-time algorithm for the problem. We found no efficient algorithms or practical implementations of the egalitarian mechanism prior to our work.

Building on~\cite{dichotomous}, Roth et al~\cite{kidney} propose an egalitarian mechanism
for the exchange of donor kidneys for transplant.
McElfresh and
Dickerson~\cite{balancing_kidney} propose a
tradeoff between fairness and a utilitarian objective function in kidney exchange programs.
Kamada and Kojima~\cite{matching_econ} study randomized matching mechanisms for
the design of matching markets under distributional constraints; their setup contains full bipartite graphs equipped with complete and strict preference relationships.
Teo and Sethuraman~\cite{geometry_matchings} prove the existence of a ``median'' deterministic solution to the stable matching problem which is fair to everyone, but finding a
polynomial-time algorithm
remains an open problem. Cheng~\cite{generalized_median}
presents a technique to approximate the median stable matching.

In the area of resource allocation problems, several works  investigate the equitable distribution of
divisible resources in networks. The work of
Ichimori et al~\cite{optimal_sharing} considers a minmax-style optimization function,
whereas Katoh et al~\cite{equipollent}  considers allocation problems so that the maximum of profit differences is minimized;
none of these consider distributions of several solutions.
Bansal et al~\cite{santa_claus} give approximation algorithms for the Santa Claus problem, where a number of indivisible presents are to be distributed among kids who have different values for
different presents, and the
goal is to maximize the minimum happiness of a kid.
Bertsimas et al~\cite{price_fairness} introduce the price of fairness in resource allocation problems. A substantial amount of work has also been devoted to cake-cutting algorithms and their strategic and incentive properties: see
~\cite{fair_division_survey,better_cake,cake_soda} and the references therein.

Several authors have studied lexicographically optimal flows in networks (which could be used in place of Step 1 of our algorithms):
Meggido~\cite{megiddo} designed an algorithm with running time $O(n^5)$, whereas Brown~\cite{sharing} proposed a polynomial-time algorithm requiring $n$ max flow computations.
On
the other hand, the parametric flow algorithm of Gallo et al~\cite{parametric_flow} can be used to find fair decompositions with a single max flow, but is not compatible with the
max flow algorithm of Goldberg~\cite{flow_unit}.
None of these methods can be used to match the runtime of our algorithm to find fair decompositions.

The
bulk of the research in the area of
algorithmic bias and fairness 
has mainly focused on avoiding discrimination against a sensitive
attribute (i.e., a protected social group)  in supervised machine
learning~\cite{fairness_awareness,FeldmanFMSV15,Corbett-DaviesP17}.
Most of this literature focuses on \emph{statistical parity}, or group-level fairness, i.e., the
difference in having a positive outcome for a random individual drawn from two different
subpopulations (e.g., men and women). Feldman et al.~\cite{FeldmanFMSV15} propose to repair
attributes 
so
as to maintain per-attribute within-group ordering while enforcing statistical parity, so that a single decision threshold applied to the
transformed attributes would result in equal success rate among the two different groups. Corbett-Davies et al.~\cite{Corbett-DaviesP17}  reformulate
algorithmic fairness as constrained optimization in the context of criminal justice:
the objective is to
maximize public safety while satisfying formal fairness constraints designed to reduce 
disparities.
Dwork et al.~\cite{fairness_awareness} provide 
examples showing that statistical parity alone is  not sufficient for fairness,
and study a randomized solution for classifiers to guarantee  that ``similar individuals are treated similarly''
in an expected sense.
The idea that more qualified individuals should be chosen preferentially is present in the work of
Joseph et al.~\cite{fairness_bandits}, who study fairness in multi-armed bandit problems.
Pedreschi et al.~\cite{PedreschiRT08} introduced the related data mining problem of discovering
discrimination practices 
    in a given dataset containing
past decisions; if such a dataset is used as training set for a machine learning model, the bias
detected can be fixed before the learning phase~\cite{KamiranC11,ZliobaiteKC11}. 
Heidari et al~\cite{moral} show that many existing definitions of algorithmic fairness, such as
 predictive value parity and equality
of odds can be viewed as instantiations of economic models of equality of opportunity.
Heidari et al~\cite{ignorance} study a welfare-based measure of fairness for risk-averse individuals, and derive an efficient
mechanism for bounding individual-level inequality. 


Finally, maxmin-fairness (in a non-distributional sense) as an 
objective is used for flow control in
networks~\cite{maxmin_allocation,data_networks}.
In the context of non-discrimination, the concept dates back at least to Rawls's theory of justice~\cite{theory_justice},
       where a ``difference principle'' is advocated whereby
           social and financial inequalities are required to be
    to the advantage of the worst-off. 
    In Rawls's distributive justice,
    social 
    measures should be designed  so as to bring the greatest benefit to the least-advantaged members
        of society, in order to maximize their prospects.

\section{Conclusions}
\label{sec:conclusions}

In this paper we study the problem of algorithmic fairness towards the elements that may or not be included in a solution of a matching problem.
This is particularly (but not exclusively) important when these elements are humans.
Towards this goal, we propose the \emph{distributional maxmin fairness} for randomized algorithms.
A series of theoretical results characterize maxmin-fair distributions 
and pave the road to our practical contribution: an exact
polynomial-time
algorithm for maxmin-fair bipartite matching, which scales to graphs with 
    millions of vertices and hundreds of millions of edges.
    We also discussed methods for the transparent and accountable real-world deployment of 
    our framework.
    \mycomment{
    Finally,  we identified a widely applicable condition (efficient optimizability) which
    guarantees the computational
    efficiency of discovering maxmin-fair distributions of solutions for general search problems.
    This result opens the door to further research on maxmin-fair algorithms:
although existence is shown,
the method implicit in
the proof may lead to suboptimal running times, calling for improved methods for specific problems.
    }

Regarding future work, it would be interesting to consider notions
 of approximate fairness intended to deal with optimization problems, where solutions may have different business
value,
possibly
unrelated to satisfaction probabilities.
The goal could be to
reach a compromise between fairness and expected business value. 
It would be desirable to be able to find approximately maxmin-fair distributions more quickly than
exact maxmin-fair distributions; we leave this as an open problem.
Another interesting question is whether our methods can be extended to handle online matching/streaming settings and/or graphs which do not fit into main memory.
Finally, future work may consider other notions of fairness for randomized algorithms for
people search, people ranking and other learning problems.

\bibliographystyle{spmpsci}      
\bibliography{references}   

\pagebreak
\appendix

\section{Proofs for Section~\ref{sec:properties}: fairness and social inequality }
\label{sec:proofs_properties}
\spara{Preliminaries.} In order to prove Theorem~\ref{thm:maxminsocial}, we need to recall some additonal facts about matroids (refer to~\cite{lawler_book} for details).
The rank function $\rho \colon 2^L \to \mathbb{N}$ of a matroid is
monotone submodular, meaning that for all $S, T\subseteq L$, it holds that
$ 0 \le \rho(S \cup T) - \rho(S) \le \rho(T) - \rho(S \cap T).$
The \emph{dual matroid} of $M$ is the matroid with ground set $L$ given by
$M^* = \{ L \setminus S \mid S \in M \}$; clearly the dual of $M^*$ is $M$ itself.
The rank function of $M^*$ is given by
$\rho^*(S) = |S| - (\rho(L) - \rho(L\setminus S))$.
      The \emph{contraction} of $M$ to the set $L
      \setminus S$ is the matroid $M/S$ with ground set $L \setminus S$ and rank function
      $\rho_{M/S}(X) = \rho(S \cup  X) - \rho(S).$
The \emph{restriction} of $M$ to the set $S$ is the matroid
      $M\rst{S}$ with ground set $S$ and independent sets $M\rst{S} = \{ I \in M \mid I \subseteq S\}$.

\begin{proof}[\textbf{Proof of Lemma~\ref{lem:unique}}]
Assume that the set $$A = \{u \in \calU \mid \cov{F}{u} \neq \cov{D}{u}\}$$ is non-empty.
Let $$u = \argmin \{\, \min(\cov{F}{u}, \cov{D}{u}) \mid u\in A\, \},$$
where ties are broken arbitrarily.
Then $\cov{D}{u} \neq \cov{F}{u}$;
suppose that $\cov{D}{u} > \cov{F}{u}$.
Then for any $v \in A$, our choice of $u$ implies that $\cov{D}{v} \ge \min(\cov{D}{v}, \cov{F}{v}) \ge
\min(\cov{D}{u}, \cov{F}{u})  =     \cov{F}{u}$;
and for any  $v\notin A$, we have  $\cov{D}{v} = \cov{F}{v}$ by definition.
In either case one of the inequalities required by condition~\eqref{eq:maxmin} fails,
so $F$ is not
maxmin-fair.
Put differently, we have shown the following implication:
$$ F \text{ is maxmin-fair } \implies \cov{D}{u} < \cov{F}{u}. $$
Similarly,
$$ D \text{ is maxmin-fair } \implies \cov{F}{u} < \cov{D}{u}. $$
    But then $F$ and $D$ cannot both be maxmin-fair.
    The only way out of this contradiction is to conclude that $A$ is empty.
\end{proof}

\begin{proof}[\textbf{Proof of Theorem~\ref{thm:lexi}}]
$\implies$ Let $F$ be maxmin-fair and consider any other distribution $D$. We need to show that
$\scov{F} \succeq \scov{D}$ (that is, $\scov{F}$ is lexicographically largest). Define
    $$A = \{u \in \calU\mid \cov{F}{u} \neq \cov{D}{u}\}.$$ If $A$ is empty, the claim is
    trivial; otherwise
    let
    $$u = \argmin \{ \cov{F}{u} \mid u\in A \} \; \text{ and } \; B = \{v \in \calU \mid \cov{F}{v} < \cov{F}{u}\}.$$

    Note that
    $u \in A\subseteq \overline{B}$
    by our choice of $u$.
    If $\cov{D}{u} > \cov{F}{u}$,
    from the maxmin-fairness of $F$ we infer the existence of $v\in A\subseteq \overline{B}$ such that
        $\cov{D}{v} < \cov{F}{u}$. This also holds if $\cov{D}{u} < \cov{F}{u}$ (then we can take $v
                = u$).
        In any case we have
    $$\min \{ \cov{D}{v} \mid v \notin B \} < \cov{F}{u} = \min \{ \cov{F}{v} \mid v \notin B \} $$
    $$ \text{ and } \qquad
     \cov{D}{v} = \cov{F}{v} < \cov{F}{u} \; \forall v \in B. $$
\mycomment{
    $$u = \argmin \{ \min(\cov{F}{u}, \cov{D}{u}) \mid u\in A \}.$$ As in the proof of
    Lemma~\ref{lem:unique}, from the maxmin-fairness of $F$ we infer
    $ \cov{F}{u} > \cov{D}{u}.$
    For any $v$ such that $\min( \cov{F}{v}, \cov{D}{v} ) < \cov{D}{u}$ we must have $v \notin A$, i.e., $\cov{F}{v} =
    \cov{D}{v}$. 
    Also, $\cov{D}{v} = \cov{D}{u}$ implies $\cov{F}{v} \ge \cov{D}{v} =
    \cov{D}{u}$ (where the first inequality holds with equality in the case $v \notin A$).
    Then for any $v\in A$ it cannot be the case that both $\cov{F}{v} < \cov{D}{u}$ and $\cov{F}{v} \le
    \cov{D}{v}$ hold, as this would imply $\min(\cov{F}{v}, \cov{D}{v}) = \cov{F}{v}   < \cov{D}{u} =
    \min(\cov{F}{u}, \cov{D}{u})$.
    In other words:
    $$ \cov{F}{u} > \cov{D}{u} \text{ and } (\cov{F}{v} \le \cov{D}{u} \implies \cov{F}{v} \ge \cov{D}{v}).$$
}
It is
readily verified that this implies $\scov{F} \succ \scov{D}$.

$\impliedby$
Let $F$ be a distribution which is not maxmin-fair. We show that $F$ is not lexicographically largest either.
Since~\eqref{eq:maxmin} does
        not hold for $F$,
there exists another distribution $D$ and a user
$u \in \calU$ such that
\begin{equation}\label{eq:unfair}
\cov{D}{u} > \cov{F}{u} \text{ and } (\cov{D}{v} < \cov{F}{v} \implies \cov{F}{v} > \cov{F}{u}) \;
\forall v.
\end{equation}
For any $\epsilon \in (0, 1)$, let $X_\epsilon$ denote the distribution picking $F$ with probability $1-\epsilon$ and $D$
    with probability $\epsilon$, so that $$\cov{X_\epsilon}{v} = \cov{F}{v} + \epsilon
    (\cov{D}{v} - \cov{F}{v})\quad \forall v.$$
    Choose $\epsilon > 0$ small enough so as to guarantee that
\begin{equation}\label{eq:less_u}
        (\cov{F}{v} < \cov{F}{u}  \implies \cov{X_{\epsilon}}{v} < \cov{X_{\epsilon}}{u}) \quad \forall v
\end{equation}
        and
\begin{equation}\label{eq:more_u}
        (\cov{F}{v} > \cov{F}{u}  \implies \cov{X_{\epsilon}}{v} > \cov{X_{\epsilon}}{u}) \quad
            \forall v.
\end{equation}
        For instance, any $$\epsilon <
        \min \left\{ \frac{|\cov{F}{u} -
            \cov{F}{v}|}{|\cov{D}{v} - \cov{F}{v}| + |\cov{D}{u} - \cov{F}{u}|}
            \Big\vert \cov{F}{v} \neq\cov{F}{u} \right\}$$
            will do.
            We have,    by~\eqref{eq:unfair},
        \begin{equation}\label{eq:increase}
(\cov{F}{v} \le \cov{F}{u}  \implies \cov{D}{v} \ge \cov{F}{v} \implies
                \cov{X_{\epsilon}}{v} \ge                 \cov{F}{v}) \quad \forall v
    \end{equation}
        \begin{equation}\label{eq:increase2}
        \text{and} \qquad
     \cov{D}{u} > \cov{F}{u}.
\end{equation}
    But~\eqref{eq:less_u},~\eqref{eq:more_u},~\eqref{eq:increase} and~\eqref{eq:increase2} say that
        $\scov{X_\epsilon}$ is
        strictly larger than $\scov{F}$ in lexicographical order, as we wished to show.
\end{proof}

The following two analogues of Lemma~\ref{lem:unique} and Theorem~\ref{thm:lexi} are also needed for the proof of Theorem~\ref{thm:maxminsocial}:
\begin{lemma}\label{lem:unique2} If $F$ and $D$ are both minmax-Pareto, then $\cov{F}{u} = \cov{D}{u}$ for all $u \in \calU$.  \end{lemma}
{
\begin{proof}
Assume that the set $$A = \{u \in \calU \mid \cov{F}{u} \neq \cov{D}{u}\}$$ is non-empty.
Let $$u = \argmax \{\, \max(\cov{F}{u}, \cov{D}{u}) \mid u\in A\, \},$$
where ties are broken arbitrarily.
Then $\cov{D}{u} \neq \cov{F}{u}$;
suppose that $\cov{D}{u} < \cov{F}{u}$.
Then for any $v \in A$, our choice of $u$ implies that $\cov{D}{v} \le \max(\cov{D}{v}, \cov{F}{v})
\le
\max(\cov{D}{u}, \cov{F}{u})  =     \cov{F}{u}$;
and for any  $v\notin A$, we have  $\cov{D}{v} = \cov{F}{v}$ by definition.
In either case one of the inequalities required by the definition of minmax-Pareto efficiency fails.
Put differently, we have shown the following implication:
$$ F \text{ is minmax-Pareto } \implies \cov{D}{u} > \cov{F}{u}. $$
Similarly,
$$ D \text{ is minmax-Pareto } \implies \cov{F}{u} > \cov{D}{u}. $$
    But then $F$ and $D$ cannot both be minmax-Pareto.
    The only way out of this contradiction is to conclude that $A$ is empty.
\end{proof}
}



\begin{theorem}\label{thm:lexi2}
For matroid problems,
a distribution $F$ is minmax-Pareto if and only if $F$ is Pareto-efficient and $\rscov{F} \preceq \rscov{D}$ for all Pareto-efficient distributions~$D$.
\end{theorem}
\begin{proof}
First observe that, for matroids, a distribution is Pareto-efficient if and only if it is supported over bases.
For any distribution $D$ of bases over a matroid $M$ with ground set $L$, consider the distribution $D^*$ of
$(L \setminus X \mid X \sim D)$
of bases over the dual matroid $M^*$. Then we have $\cov{D}{u} + \cov{D^*}{u} = 1$ for all $u \in
L$,  so clearly $F$ is minmax-Pareto if and only if $F^*$ is maxmin-fair, which (by
        Theorem~\ref{thm:lexi}) occurs if and only if $F^*$ is lexicographically largest for $M^*$, which in turn
is equivalent to $F$ being  lexicographically smallest among distributions of bases of $M$, as we wished to show.
\end{proof}

Our next result 
asserts that the only obstruction to achieving high
satisfaction probability for every user is the existence of a set of users with small rank-to-size
ratio.
Finding these obstruction sets will enable us to devise a divide and conquer strategy to obtain fair
distributions. For instance, in Example~\ref{example1} the obstruction set is given by the set of users $\{a_0, a_2, a_3\}$, which force the maximum satisfaction probability to
be no larger than $\frac{2}{3}$.
\begin{theorem}\label{thm:matroids}
Let $M$ be a matroid with ground set $L$ and rank function $\rho \colon 2^L \to \mathbb{N}$.
The minimum satisfaction probability in a minmax-fair distribution over $M$ is
$$ \pi(M) = \min \left\{ \frac{\rho(X)}{|X|} \mid {0 \neq X \subseteq L} \right\}.$$
\end{theorem}
\begin{proof}
   Any maxmin-fair distribution is supported on the collection $\mathcal{B}$ of
    bases of $M$, since extending an independent set to a base containing it never decreases any
    satisfaction probability.
Optimizing the smallest satisfaction probability  $\lambda_1$ amounts to finding a suitable distribution over
$\mathcal{B}$; let us denote the corresponding probabilities by
$\{p_B\}_{B \in \mathcal{B}}$.
Since the probability of $v \in \calU$ being included is $\sum_{v \in B} p_B$,
maximizing the minimum such probability is modeled by Program~\eqref{lp:fair} below.
It may be written as a linear
program by introducing an additional variable $\lambda$ to be maximized, and introducing the constraints $\sum_{B \ni v} p_B \ge \lambda$. Its dual is equivalent to~\eqref{lp:zv}.

{
{
\begin{minipage}{.48\linewidth}
    \begin{equation}\label{lp:fair}
    \begin{array}{rrclcl}
    \displaystyle \max & \displaystyle \min_{v\in \calU} \sum_{B \ni v} p_B \\
            \\
    \textrm{s.t.}
    & \displaystyle \sum_{B \in \mathcal{B}} p_B  & = & 1 &  \\
    & p_B    &\ge& 0 \\
    \end{array}
    \end{equation}
\end{minipage}%
\begin{minipage}{.48\linewidth}
    \begin{equation}\label{lp:zv}
    \begin{array}{rrclcl}
    \displaystyle \min & \displaystyle \max_{B\in\mathcal{B}} \displaystyle \sum_{v \in B} z_v \\
            \\
    \textrm{s.t.}
    & \displaystyle \sum_{v \in \calU} z_v  & = & 1 &  \\
    & z_v &\ge& 0. \\
    \end{array}
    \end{equation}
\end{minipage}
}
}

\noindent Observe that $\max_{B\in\mathcal{B}} \sum_{v\in\mathcal{B}} z_v$ is the value of a maximum-weight base of $M$, with weights given by
$\{z_v\}_{v \in \calU}$. Thus LP~\eqref{lp:zv} encodes the task of finding an assignment of weights to
elements of $\calU$
minimizing the maximum weight of a base. We will turn this min-max problem into a pure minimization problem.

\mycomment{
    Edmonds showed that for any fixed assignment of non-negative weights to the elements of $M$, a maximum-weight base
    may be found via the  following greedy algorithm:
    \begin{enumerate}
        \item Sort the weights by decreasing order: $z_{\pi_1} \ge z_{\pi_2} \ge \ldots z_{\pi_n}$, where
              we assume $L = [n]$ and $\pi \in \mathcal{S}(L)$ is a suitable permutation of $L$.
        \item Set $B \gets \emptyset$.
        \item For $i = 1 \ldots n$, add $\pi_i$ to $B$ if $B \cup \{\pi_i\}$ is independent in $M$.
    \end{enumerate}
}

Edmonds~\cite{matroids_greedy} showed that for any fixed assignment of non-negative weights to the elements of~$M$, a maximum-weight base
may be found via the greedy algorithm that examines each element in order of decreasing weight and
adds it to the
current set if its addition does not violate independence.
Let $\Pi$ denote the set of permutations of $\calU = \{1, 2, \ldots, n\}$. Write $$\Delta = \left\{z \in \mathbb{R}^\calU
\big\vert
\sum_{v \in \calU} z_v = 1, \quad z_v \ge 0 \; \forall v \in \calU\right\}$$ for the
probability simplex on $\calU$ and let
$G(\pi) = \{ z \in \Delta \mid z_{\pi(1)} \ge z_{\pi(2)} \ge \ldots \ge z_{\pi(n)} \}$ denote the elements of $\Delta$
which become sorted after applying permutation $\pi \in \Pi$.

Note that if $z, z' \in G(\pi)$, then the two bases obtained via the greedy algorithm with vertex weights $\{z_v\}$ and
$\{z'_v\}$ are the same. For each
$\pi \in \Pi$, let $B(\pi)$
denote the base obtained via the greedy algorithm; Edmond's result may then be written as
$$ \max_{B \in \mathcal{B}} \sum_{v \in B} z_v = \sum_{v \in B(\pi)} z_v \quad \text{ if } z \in G(\pi). $$

By LP~\eqref{lp:zv}, the fairness parameter $\lambda_1$ is
\begin{equation}
\label{eq:perm2}
\min_{z \in \Delta}\; \max_{B \in \mathcal{B}} \sum_{v \in B} z_v
    = \min_{\pi \in \Pi}\, \min_{z \in G(\pi)}\, \max_{B \in \mathcal{B}} \sum_{v \in B} z_v  = \min_{\pi \in \Pi}\, \min_{z \in G(\pi)}\, \sum_{v \in B(\pi)} z_v.
\end{equation}
We claim that, for each $\pi \in \Pi$ and each non-empty $X \subseteq \calU$,
\begin{equation}\label{eq:perm}
    \min_{z \in G(\pi)} \sum_{v \in X} z_v = \min_{i \in [n]} \frac{|X \cap \pi([i])|}{i} ,
\end{equation}
where $\pi([i]) = \{ \pi(1), \ldots, \pi(i) \}$.

    This means that,
    if we are given advice on the permutation $\pi$ which sorts an optimal solution $z$ to LP~\eqref{lp:zv}, then we can find another solution $\hat{z}$ with the same value and whose non-zero
    weights are evenly distributed among the top $t$ elements of $z$ in sorted order, for some $t \in [n]$. 
For some optimal $t$, 
each of the $t$ non-zero values of $\hat{z}_i$ is either $0$ or $1/t$. 
To see this assuming~\ref{eq:perm}, notice that we can construct such $\hat{z}$ by setting $\hat{z}_{\pi(i)} = \frac{1}{t}$ for $1 \le i \le t$ and $\hat{z}_j = 0$ for $j \notin \pi([t])$.

To see why~\eqref{eq:perm} holds, define $d_{n} = z_{\pi(n)}$ and $d_i = z_{\pi(i)} - z_{\pi(i+1)}
\ge 0$ for $0 < i < n$.
Then $z_{\pi(i)} = \sum_{j \ge i} d_i,$ hence
$$ \sum_{v \in X} z_v = \sum_{i\in[n]} \left( \mathbbm{1}[{\pi(i) \in X}] \cdot \sum_{j \ge i} d_i
        \right) =  \sum_{j\in[n]} d_j \cdot |X \cap
\pi([j])|.$$
The conditions $\sum_{v \in \calU} z_v = 1$  and $z \in G(\pi)$ then  become $\sum_i i \cdot d_i = 1$
and $d_i \ge 0$.
Therefore
$$   \min_{z \in G(\pi)} \sum_{v \in X} z_v = \min \left\{ \sum_{i\in[n]} d_i \cdot |X \cap \pi([i])|
\;\Big\vert\;
\sum_{i\in[n]} i \cdot d_i = 1, d_i \ge 0 \right\}.
$$
The quantity in the right-hand side equals the smallest ratio (among all $i$) between the coefficient of $d_i$ in
the objective function  ($|X \cap \pi([i])|$) and in the only equality constraint ($i$),
    proving~\eqref{eq:perm}. From~\eqref{eq:perm2} and~\eqref{eq:perm} the theorem follows, because the greedy algorithm satisfies
$|B(\pi) \cap \pi([i])| = \rho(\pi[i])$ for all $i$, so if $\displaystyle S^*= \argmin_{S\subseteq
    \calU}\frac{ \rho(S) }{ |S|}$, then for any $\pi \in \Pi$ we have
$$
     \min_{z \in G(\pi)} \sum_{v \in B(\pi)} z_v =\min_{i \in [n]} \frac{|B(\pi) \cap \pi([i])|}{i}
 = \min_{i \in [n]}
\frac{\rho(\pi([i]))}{|\pi([i])|} \ge \frac{\rho(S^*)}{|S^*|},
$$
and equality holds for any permutation where the elements of $S^*$ precede those of $\calU \setminus
S^*$. 
\mycomment{
    The expression for $\lambda_{|L|}$ can be obtained by observing that  $1-
    \lambda_{|L|}$ is the minimum satisfaction probability in a maxmin-fair distribution for the
    \emph{dual} matroid of $M$, that is,
     the matroid $M^*$ with ground set $L$ given by
    $M^* = \{ L \setminus S \mid S \in M \}$, whose rank function $\rho^*$ is given by
    $\rho^*(S) = |S| - (\rho(L) - \rho(L\setminus S))$.
}
\end{proof}
\begin{theorem}\label{coro:matroids2}
Let $M$ be a matroid with ground set $L$ and rank function $\rho \colon 2^L \to \mathbb{N}$.
The maximum satisfaction probability in a minmax-Pareto distribution over $M$ is
$$ \Pi(M) = \max \left\{ \frac{\rho(L) - \rho(X)}{|L \setminus X|} \mid {0 \subseteq X \subsetneq L} \right\}.$$
\end{theorem}
\begin{proof}
Apply Theorem~\ref{thm:matroids} to the dual matroid of $M$.
\end{proof}

An extension of Theorem~\ref{thm:matroids} allows us to compute the satisfaction probability of every element of $L$.
\begin{lemma}\label{lem:matroids_ext}
Define a sequence of sets $B_1, B_2, \ldots, B_k$ iteratively by:
\begin{equation}\label{eq:rule}
 B_i \text{ is a maximal set } X \subseteq L\setminus S_{i-1} \text{ minimizing } \frac{ \rho(X \cup S_{i-1}) - \rho(S_{i-1}) }{|X|}
\text{, where $S_i = \bigcup_{j \le i} B_i$.}
\end{equation}
We stop when $S_i = L$ (which will eventually occur as the sequence $\{S_i\}$ is strictly increasing).     
Then for every $i, u \in B_i$, the satisfaction probability of $u$ in a maxmin-fair distribution $F$ is $\lambda_i = \frac{ \rho (B_i) }{ |B_i| }$.
\end{lemma}
(Maximality of each $B_i$ is not required for the conclusion to hold, but its inclusion guarantees uniqueness of the sets thus defined, owing to the submodularity of $\rho$.)
\begin{proof}
We reason by induction on the number $k$ of sets.
   First, observe that Theorem~\ref{thm:matroids} implies $F[u] \ge \lambda_1$ for all $u \in L$. As the expected number of satisfied elements within $B_1$, which obviously
           cannot exceed $\rho(B_1) = \lambda_1 |B_1|$,
           is equal to
\begin{equation}\label{eq:exp_covered}
\expect_{A \sim F} [|A \cap B_1|] =  \sum_{u\in B_1} F[u] \ge \lambda_1 |B_1|
\end{equation}
by linearity of expectation, the equality $F[u] = \lambda_1$ must hold for all $u \in B_1$.
If $k = 1$, this shows the result. 

If $k > 1$, let $D_1$ be a maxmin-fair distribution for the restriction $M \rst{B_1}$ of $M$ to $B_1$
and let $D_2$ be a maxmin-fair distribution for the contraction $M/(L\setminus B_1)$ of $M$ to the remaining elements $L \setminus B_1$.
Since restriction does not change the rank
function within $B_1$, $D_1$ satisfies $D_1[u] = F[u] = \lambda_1$ for all $u \in B_1$. The rank function of the contraction $M/(L\setminus B_1)$ is  $\rho_{M/B_1}(X) = \rho(X \cup B_1) -
\rho(B_1)$, so by applying rule~\eqref{eq:rule} iteratively we obtain the same sequence of sets $B_2, \ldots, B_k$ (excluding $B_1$), and by the induction hypothesis
    $D_2$ satisfies $D_2[u] = \lambda_i$ for all $i\ge 2, u \in B_i$. 
    It remains  to  be shown that $D_2[u] = F[u]$ for all $u \notin B_1$.

    Denote by $[D_1 \cup D_2]$
denote 
the distribution of
    $ (A \cup B \mid A \sim D_1, B \sim D_2)$. This is a distribution over independent sets of $M$
    by the following property of matroid contractions (see~\cite{lawler_book}):
\begin{align}\label{eq:contraction}
\text{for any base $B$ of $M\rst{B_1}$, }& \text{a subset $I \subseteq L \setminus B_1$ is independent in $M/B_1$}\nonumber \\
&\text{if and only if $I \cup B_1$ is independent in $M$.}
\end{align} 
On the other hand, for any set  $A$ in the support of a maxmin-fair distribution $F$, the set $A \cap B_1$ must be a base of
    $M\rst{B_1}$ (or else Equation~\eqref{eq:exp_covered} would fail). Let $F_2$ denote the distribution $(A \setminus B_1 \mid A \sim F)$; by~\eqref{eq:contraction}, $F_2$ is
    a distribution over elements of the contraction $M/B_1$.

To complete the proof, observe that
$F\uparrow \succeq [D_1 \cup D_2]\uparrow$ by Theorem~\ref{thm:lexi} because $F$ is maxmin-fair.
As $\cov{F}{u} = \cov{D_1}{u} <
\cov{F}{v}, \cov{D}{v}$ for $u \in B_1, v \in L\setminus B_1$, the fact that $F\uparrow \succeq [D_1
\cup D_2]\uparrow$ implies $F_2\uparrow \succeq D_2\uparrow$, and the maxmin-fairness of $D_2$ allows us to deduce that 
$F_2\uparrow  = D_2\uparrow$. Hence, by
Lemma~\ref{lem:unique}, for all $v \notin B_1$ we have
$\cov{F}{v} = \cov{F_2}{v} =
\cov{D_2}{v}$.

\end{proof}

    Similarly, we have the following for minmax-fairness.  
    \begin{lemma}\label{lem:matroids_ext2}
Define a sequence of sets $B'_1, \ldots, B'_{k'}$ iteratively by:
\begin{equation}\label{eq:rule2}
 B'_i \text{ is a maximal set } X \subseteq L\setminus S_i \text{ maximizing } \frac{ \rho(S'_{i}) - \rho(S'_{i}\setminus X) }{|X|}
\text{, where $S'_i = L\setminus \bigcup_{j < i} B'_i$.}
\end{equation}
We stop when $S'_k = \emptyset$.
Then for every $i, u \in B'_i$, the satisfaction probability of $u$ in a minmax-Pareto distribution $F$ is $\lambda'_i = \frac{ \rho (B'_i) }{ |B'_i| }$.
\end{lemma}
\begin{proof}
We argue by induction on $k$.
   First, observe that Theorem~\ref{coro:matroids2} implies $F[u] \le \lambda_1$ for all $u \in L$. 
   The expected number of satisfied elements within $B_1$ cannot be below
           $\lambda_1 |B_1|$ for any Pareto-efficient distribution $F$, otherwise 
           we would have the contradiction
$$  \rho(L) = \expect_{A \sim F} |A| =  \expect_{A \sim F} [|A \cap B_1|] + \expect_{A \sim F} [|A \setminus B_1|] 
                                     <  \lambda_1 |B_1| +  \rho(L\setminus B_1) = \rho(L) .
                                     $$
On the other hand,
\begin{equation}\label{eq:exp_covered2}
\expect_{A \sim F} [|A \cap B_1|] =  \sum_{u\in B_1} F[u] \le \lambda_1 |B_1| = \rho(B_1)
\end{equation}
by linearity of expectation, so the equality $F[u] = \lambda_1$ must hold for all $u \in B_1$.
If $k = 1$, this shows the result. The rest of the proof is completely analogous to that of  
Lemma~\ref{lem:matroids_ext}, except that we use Theorem~\ref{thm:lexi2} and Lemma~\ref{lem:unique2} in place of Theorem~\ref{thm:lexi} and Lemma~\ref{lem:unique}.
\end{proof}

\begin{proof}[\textbf{Proof of Theorem~\ref{thm:maxminsocial}}]
It suffices to prove the equivalence
(1) $\Leftrightarrow$ (2). Inded, if it holds, 
then
    a maxmin-fair distribution simultaneously maximizes the minimum satisfaction probability and
    minimizes the maximum satisfaction probability (among Pareto-efficient distributions), hence it also minimizes the largest difference between two
    satisfaction probabilities. An easy inductive argument (omitted) shows that
the equivalence (1) $\Leftrightarrow$ (3) then follows.

To show that (1) $\Leftrightarrow$ (2), consider the sequence $B_1, \ldots, B_k$ from Lemma~\ref{lem:matroids_ext} and the sequence 
$B'_1, \ldots, B'_{k'}$ from Lemma~\ref{lem:matroids_ext2}. It suffices to show that they are the same sequence in reverse: $k = k'$ and $B_i = B'_{k+1-i}$ for all $i$.
We proceed from  top to bottom, showing
   by induction that for each value of $i$ from 1 to $k$, $B'_i = B_{k+1-i}$.
Consider any $Z \in S_{i+1} = S_i \cup B_i$, which may be split into $Z = X \cup Y$ where
$X \subseteq S_i = \bigcup_{j < i} B_i$ and $Y \in B_i$. 
Then we have
$$ \rho(Z) - \rho(S_i) \ge \rho(S_i \cup Y) - \rho(S_i) \ge \lambda_i |Y|, $$
    where the first inequality is by submodularity of $\rho$, and the second by construction of $B_i$.
    On the other hand,
$$ \rho(S_{i+1}) - \rho(S_i) = \lambda_i |B_i|, $$
hence
$$\frac{ \rho(S_{i+1}) - \rho(Z) }{|S_{i+1} \setminus Z|} \ge
\frac{\lambda_i |B_i| + \rho(S_i) - (\lambda_i |Y| + \rho(S_i))  }{|B_i \setminus Y|} = \lambda_i.
$$
Notice that equality holds when $Z = B_i$, so $B_i$ maximizes 
$\frac{ \rho(S_{i+1}) - \rho(Z) }{|S_{i+1} \setminus Z|}$ over all $Z \subseteq S_{i+1}$.
Using the definition of $B'_{i'}$, this means that $B'_{i'} = B'_{i}$, as we wished to show.
\end{proof}

\section{Proofs for Section~\ref{sec:problem}: a polynomial-time algorithm for maxmin-fair matching}
\label{sec:proofs_problem}

\begin{proof}[\textbf{Proof of Corollary~\ref{thm:lambda2}}]
 From Theorem~\ref{thm:hall} it follows that, when $\rho(L) = |R|$,  the rank function of a bipartite matching problem is given by
         \begin{equation}\label{eq:hall_rank}
         \rho(S) = \min_{T \subseteq S} |\Gamma(T)| + |S| - |T|.
         \end{equation}
         Define
         \begin{equation*}
      \alpha = \max_{S\subsetneq L} \frac{|\Gamma(L)| - |\Gamma(S)|}{|L \setminus S|}; \qquad 
      \beta  = \max_{S\subsetneq L} \frac{\rho(L) - \rho(S)}{|L \setminus S|} .
         \end{equation*}
     In view of Corollary~\ref{coro:matroids2} and the equivalence between maxmin and minmax fairness for matroid problems (Theorem~\ref{thm:maxminsocial}),
         it suffices to show that $\alpha = \beta$.
     Since $\rho(L) = |\Gamma(L)|$ and $\rho(S) \le |\Gamma(S)|$ for all $S$, inequality $\alpha \le \beta$ is immediate. To show that $\beta \le \alpha$, it suffices
     to prove that
$ \rho(L) - \rho(S) - \alpha |L \setminus S| \le 0$ for
         all $S \subseteq L$. This follows from~\eqref{eq:hall_rank} and the fact that $\alpha\le\beta\le 1$:
                \begin{align*}
\rho(L) - \rho(S) - \alpha|L \setminus S| &=   \max_{T \subseteq S}{\; \Gamma(L) - \Gamma(T) - |S| + |T|- \alpha|L \setminus S|}\\
                             &\le  \max_{T \subseteq S}{\; \alpha|L \setminus T| - |S \setminus T|- \alpha|L \setminus S|}\\
                             &\le 0.
             \end{align*}
\end{proof}

\mycomment{
\begin{proof}[\textbf{Proof of Lemma~\ref{lem:combine}}]
The proof is only sketched as it closely follows the argument in case (b) in the proof of
Theorem~\ref{thm:maxminsocial}; in fact the notions of restriction and contraction just defined
coincide with those for matroids.
For any maximum matching $B$ in $G\rst{X}$, a subset $I \subseteq L \setminus X$ is matchable in
$G/S$ if and only if $B \cup I$ is matchable in $G$.
For any distribution $X$ over $L$, there exists a
distribution of the form $[X_1 \cup X_2]$
    with the same satisfaction
    probabilities
    for all $u \in L$.
In particular some maxmin-fair distribution $F$ for $G$ can be written in this form: $F = [F_1 \cup
F_2]$, where $F_1$ is a distribution over $G\rst{X}$ and $F_2$ is a distribution over $G/X$. As in
the proof of Theorem~\ref{thm:maxminsocial}, it follows from the fair isolation of $X$ that $F_1$ and $F_2$ are maxmin-fair,
    implying (b) by the uniqueness of satisfaction probabilities of maxmin-fair distributions; and
    conversely, if $F_1$ and $F_2$ are maxmin fair, then $[F_1 \cup F_2]$ is
    maxmin-fair, proving (a).

\end{proof}
}

\begin{lemma}\label{lem:graphs_ext}
Define a sequence of sets $B_1, B_2, \ldots, B_k$ iteratively by:
\begin{equation}\label{eq:rule}
 B_i \text{ is a maximal set } X \subseteq L\setminus S_{i-1} \text{ minimizing } \frac{ |\Gamma(X \cup S_{i-1})| - |\Gamma(S_{i})| }{|X|}
\text{, where $S_i = \bigcup_{j <i} B_i$.}
\end{equation}
Stop when $S_k = L$.
Then for every $i, u \in B_i$, the satisfaction probability of $u$ in a maxmin-fair distribution $F$ is $\lambda_i = \frac{ \rho (B_i) }{ |B_i| }$, 
    and any $w \in \Gamma(B_i)\setminus \Gamma(S_{i-1})$ is matched to some $u \in B_i$ with probability 1.
\end{lemma}
\begin{proof}
Since the sequence $S_0, S_1, \ldots$ is strictly increasing (with respect to inclusion) and $L$ is
finite, there exists some $k$ such that $S_k = L$.

For each $i=0,\ldots, k$, let $H_i$ denote the graph $(G/S_{i-1})\rst{L\setminus S_{i-1}}$, i.e., the result of removing the vertices in $S_{i-1}$ and all their incident edges.
For $i = 1, \ldots, k$, we argue by induction on $i$ that
the coverage probabilities of $F$ outside of $S_{i-1}$ coincide with those of a maxmin-fair distribution for  $H_i$;
and
and moreover the probabilities are as prescribed by the statement of the
lemma.

The case $i = 1$ is trivial, so assume $i > 1$. 
By Corollary~\ref{thm:lambda1},
   there is a distribution of matchings in $H_i$ with minimum
satisfaction probability at least $\lambda_i$; the expected number of covered elements from $B_i$ is
then at least $\lambda_i |B_i| = |\Gamma(B_i) \setminus \Gamma(S_{i-1})| = |\Gamma_{H_i}(B_i)|$. Hence
equality must always hold, and the maxmin-fair distribution $F_i$ for $H_i$ has satisfaction probability
precisely
$\lambda_i$ for all $u \in B_i$. 
By the induction hypothesis, $F[u] = F_i[u] = \lambda_i$ for all $u \in B_i$.
Now observe that the neighbors of $B_i$ that belong to $S_{i-1}$ are already matched with probability 1. There are only 
$|\Gamma_{H_i}(B_i)|$ other neighbors, and
since the expected number of covered neighbours of $B_i$ in $F$ is equal to $|\Gamma_{H_i}(B_i)|$,
    it follows that any $w \in \Gamma_{H_i}(B_i)$ is matched to some $v
\in B_i$ with probability 1 in $F$.
In particular, in $F$ no element of $\Gamma_{H_i}(B_i)$ is matched to any vertex outside
        $B_i$ with non-zero probability, so the satisfaction probabilities of $F$ outside of $S_i$ must coincide with those of of a maxmin-fair distribution for $H_{i+1}$.
\end{proof}

\mycomment{
\begin{lemma}\label{lem:matroids_ext}
Define a sequence of sets $B_1, B_2, \ldots, B_k$ iteratively by:
\begin{equation}\label{eq:rule}
 B_i \text{ is a maximal set } X \subseteq L\setminus S_i \text{ minimizing } \frac{ \rho(X \cup S_i) - \rho(S_{i}) }{|X|}
\text{, where $S_i = \bigcup_{j <i} B_i$.}
\end{equation}
We stop when $S_k = L$ (which will eventually occur as the sequence $\{S_i\}$ is strictly increasing).     
Then for every $i, u \in B_i$, the satisfaction probability of $u$ in a maxmin-fair distribution is
    $\lambda_i = \frac{|\Gamma(B_i) \setminus \Gamma(S_{i-1} |)}{|B_i|}$,
    and any $w \in \Gamma(B_i)\setminus \Gamma(S_{i-1})$ is matched to some $u \in B_i$ with probability 1. 
\end{lemma}
(Maximality of each $B_i$ is not required for the conclusion to hold, but its inclusion guarantees uniqueness of the sets thus defined, owing to the submodularity of $\rho$.)
\begin{proof}
\end{proof}
}

\begin{lemma}\label{lem:isol_chain}
For any two distinct fairly isolated sets $X$ and $Y$, either $X \subseteq Y$ or $Y \subseteq X$ holds.
\end{lemma}
\begin{proof}
By Corollaries~\ref{thm:lambda1} and~\ref{thm:lambda2}, $X \neq \emptyset $ is fairly isolated if and only if $X = L$ or
    $$\Pi(G\rst{S}) = \max_{S \subsetneq X} \frac{ |\Gamma(X)| - |\Gamma(S)| }{|X\setminus S|} < \pi(G/S) = \min_{T \supsetneq X} \frac{ |\Gamma(T)| - |\Gamma(X)| }{|T\setminus X|} .$$
For any two sets $A, B$ such that $A \subsetneq B$, define
$d(A \mid B) = \frac{ |\Gamma(A \cup B)| - |\Gamma(B)| }{|A \setminus B|}. $
Then we can rewrite the definition of fair isolation (including the case $X = L$) as:
    $$ X\text{ is fairly isolated} \qquad \Leftrightarrow \qquad d(X \mid S) < d(T \mid X) \qquad \forall\; S\subsetneq X, T\supsetneq X. $$

Now assume for contradiction $X$ and $Y$ are fairly isolated but $X \setminus Y$ and $Y\setminus X$ are both non-empty. 
Then $d(X \mid X \cap Y)$ and $d(Y\mid X \cap Y)$ are well defined;
assume without loss of generality that $d(X \mid X \cap Y) \le d(Y\mid X \cap Y)$. Then
$$ d(X\mid X \cap Y) \le d(Y\mid X \cap Y) < d(X \cup Y\mid Y),$$
where we used the fair isolation of $Y$. But this contradicts the fair isolation of $X$.
\end{proof}

    \begin{proof}[\textbf{Proof of Theorem~\ref{thm:decomp}}]

Let $B'_1, \ldots, B'_k$ be the sequence of sets given by Lemma~\ref{lem:graphs_ext} and define $S'_i = \bigcup_{j \le i} B'_i,$
$\lambda'_i = \frac{\Gamma(S'_i \cup B'_i) - \Gamma(S'_i) }{|B'_i|}$ and $\lambda'_0 = 0$.
We show that $S'_1, \ldots, S'_{k}$ 
   comprise all fairly isolated sets. Assuming this for the moment, notice that by definition these sets form a chain, and the sets $B'_1, \ldots, B'_k$ satisfy property (a) by
   Lemma~\ref{lem:graphs_ext}. Part (b) follows then by applying the fair decomposition to the dual of the
matching matroid, using Corollary~\ref{thm:lambda2}, and recalling that maxmin-fairness and minmax-Pareto efficiency are equivalent for
matroids (Theorem~\ref{thm:maxminsocial}).

To see that $S'_1, \ldots, S'_{k-1}$ are fairly isolated, notice that
$S'_k = L$ indeed is by definition, whereas for $i < k$ we have
$\Pi(G\rst{S'_i}) = \max_{u \in S_i} F[u] = \lambda'_i < \lambda'_{i+1} = \pi(G/{S'_i}).$
This meets the definition of fair separation from Section~\ref{sec:fair_decomp}.

The fact that the fairly isolated sets form a chain is a direct consequence of Lemma~\ref{lem:isol_chain}. 
Finally, assume for contradiction that some fairly isolated set $X$ exists
other than $S'_1, \ldots, S'_k$. Then
    $S'_i \subsetneq X \subsetneq S'_{i+1}$
for  some $i, 0 \le i < k$.
Then we have
\begin{equation}\label{eq:weird}
    \lambda'_{i+1} \le \frac{ |\Gamma(X)| - |\Gamma(S'_i)| }{|X\setminus S'_i|} <  \frac{ |\Gamma(S'_{i+1})| - |\Gamma(X)| }{|S'_{i+1}\setminus X|} ,
\end{equation}    
    where the first inequality is by
construction of $B'_{i+1}$ and $S'_{i+1}$, and the second by the fair isolation of $X$.
But then
    $$ \lambda'_{i+1} {|S'_{i+1}\setminus X|} + |\Gamma(X)| <  |\Gamma(S'_{i+1})| =|\Gamma(S'_{i})| + \lambda'_{i+1} |B'_i|   ,$$
    i.e.,
    $$ |\Gamma(X)| - |\Gamma(S'_{i})|< \lambda'_{i+1} {|X\setminus S'_{i}|},$$
    contradicting~\eqref{eq:weird}.
\end{proof}

\mycomment{
    The following result gives alternative equivalent definitions and will be used in the proof of Theorem~\ref{thm:decomp}:
    \begin{restatable}{lemma}{lemequivfair}
        \label{lem:equiv_fair}
        The following are equivalent for every $X \subseteq L$:
        \begin{enumerate}[(a)]
            \item $X$ is fairly isolated;
            \item ($u, v \in Y$, $F$ is maxmin-fair, and $\cov{F}{v} \le \cov{F}{u}$) $\implies$ $v \in X$;
            \item Either
        $\Gamma(X) = \emptyset$, or $X = L$, or
    $$\min\left(1, \max_{Y \subsetneq X} \frac{ |\Gamma(X)| - |\Gamma(Y)| }{|X\setminus Y|} \right)  < \min\left(1, \min_{Z \supsetneq X} \frac{ |\Gamma(Z)| - |\Gamma(X)| }{|Z\setminus
            X|} \right).$$
        \end{enumerate}
        \end{restatable}
}

\end{document}